\documentclass[11pt]{article}
\usepackage{hyperref}
\usepackage{times}  % Fonts

\usepackage{mathpazo}
\usepackage{amssymb,amsmath,amsthm}
\usepackage{tikz}
\usetikzlibrary{positioning}

\newtheorem{theorem}{Theorem}[section]
\newtheorem{proposition}[theorem]{Proposition}
\newtheorem{definition}[theorem]{Definition}

\newtheorem{claim}[theorem]{Claim}
\newtheorem{lemma}[theorem]{Lemma}

\newtheorem{corollary}[theorem]{Corollary}

\newtheorem{remark}[theorem]{Remark}

\newcommand{\qedsymb}{\hfill{\rule{2mm}{2mm}}}
\renewenvironment{proof}[1][]{\begin{trivlist}
\item[\hspace{\labelsep}{\bf\noindent Proof#1:\/}] }{\qedsymb\end{trivlist}}

\def\R{\mathbb{R}}

\def\N{\mathbb{N}}

\newcommand{\NP}{\mathsf{NP}}

\renewcommand{\P}{\mathsf{P}}

\newcommand{\YES}{\mathsf{YES}}
\newcommand{\NO}{\mathsf{NO}}

\newcommand{\od}{\overline{\xi}}

\newcommand{\eps}{\epsilon}
\renewcommand{\epsilon}{\varepsilon}

\newcommand{\rank}{\mathop{\mathrm{rank}}}

\begin{document}

\title{{\bf New Hardness Results for Low-Rank Matrix Completion}}

\author{
Dror Chawin\thanks{The Academic College of Tel Aviv-Yaffo, Tel Aviv 61083, Israel. Research supported by the Israel Science Foundation (grant No.~1218/20).}
\and
Ishay Haviv\footnotemark[1]
}

\date{}

\maketitle

\begin{abstract}
The low-rank matrix completion problem asks whether a given real matrix with missing values can be completed so that the resulting matrix has low rank or is close to a low-rank matrix. The completed matrix is often required to satisfy additional structural constraints, such as positive semi-definiteness or a bounded infinity norm. The problem arises in various research fields, including machine learning, statistics, and theoretical computer science, and has broad real-world applications.

This paper presents new $\NP$-hardness results for low-rank matrix completion problems. We show that for every sufficiently large integer $d$ and any real number $\eps \in [ 2^{-O(d)},\frac{1}{7}]$, given a partial matrix $A$ with exposed values of magnitude at most $1$ that admits a positive semi-definite completion of rank $d$, it is $\NP$-hard to find a positive semi-definite matrix that agrees with each given value of $A$ up to an additive error of at most $\eps$, even when the rank is allowed to exceed $d$ by a multiplicative factor of $O (\frac{1}{\eps^2 \cdot \log(1/\eps)} )$. This strengthens a result of Hardt, Meka, Raghavendra, and Weitz (COLT,~2014), which applies to multiplicative factors smaller than $2$ and to $\eps$ that decays polynomially in $d$. We establish similar $\NP$-hardness results for the case where the completed matrix is constrained to have a bounded infinity norm (rather than be positive semi-definite), for which all previous hardness results rely on complexity assumptions related to the Unique Games Conjecture. Our proofs involve a novel notion of nearly orthonormal representations of graphs, the concept of line digraphs, and bounds on the rank of perturbed identity matrices.
\end{abstract}

\section{Introduction}

In the matrix completion problem, the input is a partially observed real matrix, where some entries are marked by $\perp$ to indicate missing values. The goal is to fill in these values in such a way that the completed matrix satisfies certain prescribed properties. This algorithmic task is prevalent in various research fields, including machine learning, statistics, and theoretical computer science, and is motivated by a range of real-world applications, such as recommendation systems, medical imaging, computer vision, and signal processing. Typically, the completed matrix is required to have low rank or to be a slight perturbation of a low-rank matrix. In certain cases, the matrix may also need to be positive semi-definite or constrained by a bounded infinity norm.

A central objective in the study of low-rank matrix completion is to identify the conditions under which the problem can be solved efficiently. A successful line of work, initiated by Cand{\`{e}}s and Recht~\cite{CandesRecht09}, has applied convex optimization methods to efficiently recover a low-rank matrix from the values of a subset of its entries (see also~\cite{CandesTao10,Recht11}). The recovery is guaranteed to succeed if (a) the number of given entries is sufficiently large, (b) the completed matrix satisfies a certain incoherence condition (i.e., the row and column spaces of the matrix are not aligned with the vectors of the standard basis), and (c) the subset of exposed entries is drawn uniformly at random. In fact, these conditions further ensure that the solution is unique. However, in most applications, the given entries cannot be chosen at random. It is therefore of interest to determine the computational complexity of recovering a low-rank incoherent matrix, when the observed entries are selected in a worst-case manner.

The hardness of low-rank matrix completion can be traced back to a 1996 paper by Peeters~\cite{Peeters96}, who explored the complexity of determining a graph quantity known as minrank, introduced by Haemers~\cite{Haemers81} in the context of the Shannon capacity of graphs. Peeters' work implies that for every integer $d \geq 3$, deciding whether a given partial matrix can be completed to a matrix of rank at most $d$ is $\NP$-hard. His proof technique further implies that the same $\NP$-hardness result holds when the completed matrix is required to be positive semi-definite, and this was extended to the case of $d=2$ in 2013 by Eisenberg-Nagy, Laurent, and Varvitsiotis~\cite{Nagy13}. This research avenue was then extended by Hardt, Meka, Raghavendra, and Weitz~\cite{HMRW14}, who explored two relaxations of the problem: first, allowing some slackness in the rank of the completed matrix, and second, permitting a bounded additive error on the given entries of the partial matrix. Specifically, they proved that for every integer $d \geq 6$, given a partial matrix $A$ whose observed entries have magnitude at most $1$, it is $\NP$-hard to distinguish between the case where $A$ admits a positive semi-definite completion of rank at most $d$, and the case in which any positive semi-definite completion of $A$ has rank at least $2d$. Moreover, they showed that the problem remains $\NP$-hard, when the matrix in the latter case is allowed to approximate each given value up to an additive error of $\eps$, as long as $\eps$ decays polynomially with the desired rank, namely, for $\eps = O(d^{-6})$. More recently, the paper~\cite{ChawinH23} addressed the related problem of determining a graph measure known as the orthogonality dimension. The results of~\cite{ChawinH23} imply that for every sufficiently large integer $d$, it is $\NP$-hard to decide whether an input partial matrix can be completed to a positive semi-definite matrix of rank at most $d$, or any positive semi-definite completion has rank at least $2^{(1-o(1)) \cdot d/2}$, with the $o(1)$ term tending to $0$ as $d$ tends to infinity. It thus follows that it is $\NP$-hard to approximate to within any constant factor the minimum possible rank of a positive semi-definite matrix that agrees with an input partial matrix. However, this hardness result does not apply to the more tolerant setting that allows an additive perturbation in each entry of the partial matrix.

Another version of the low-rank matrix completion problem considered by Hardt et al.~\cite{HMRW14} imposes a fixed bound on the infinity norm of the completed matrix (rather than requiring positive semi-definiteness). For this setting, their hardness results were not based on the standard assumption $\P \neq \NP$, but on the hardness of appropriate gap versions of the Coloring and Independent Set problems, the intractability of which is supported by certain variants of the Unique Games Conjecture~\cite{DinurMR06,DinurS10}. Under these assumptions, they showed that for all positive integers $d_2 > d_1 \geq 3$ and real numbers $\eps \in [0,\frac{1}{2})$ and $\theta \geq 1$, there is no efficient algorithm to decide whether an input partial matrix $A$ can be completed to one with an infinity norm at most $\theta$ and rank at most $d_1$, or any completion of $A$ with an infinity norm at most $\theta$ must have rank at least $d_2$, even when an additive error of $\eps$ is allowed in each entry. Remarkably, this hardness result persists when a constant fraction of the matrix entries is revealed (as is the case for all the above hardness results) and, in addition, when the completed matrix is required to be incoherent (in instances admitting a valid completion). In light of the aforementioned algorithmic results for the problem~\cite{CandesRecht09,CandesTao10,Recht11}, these findings highlight the worst-case choice of the revealed entries as a substantial obstacle to efficient recovery.

\subsection{Our Contribution}

This paper presents several new hardness results for low-rank matrix completion problems. We begin by considering the case where the completed matrix is required to be positive semi-definite. Specifically, we study the gap problem $(d_1,d_2,\eps)$-PSD-Completion, formally defined below.
Here, we let $\mu(B)$ denote the coherence of a matrix $B$, a measure that is always bounded from below by $1$ (see Definition~\ref{def:coherence}).

\begin{definition}[The $(d_1,d_2,\eps)$-PSD-Completion Problem]\label{def:PSD_prob}
For positive integers $d_1 < d_2$ and for a real number $\eps \in [0,1)$, the {\em $(d_1,d_2,\eps)$-PSD-Completion} problem asks, given a partial matrix
\[A \in ([-1,+1] \cup \{\perp\})^{n \times n},\] to distinguish between the following cases.
\begin{itemize}
  \item $\YES$: There exists a positive semi-definite matrix $B \in \R^{n \times n}$, such that $A_{i,j} = B_{i,j}$ for all $i,j \in [n]$ with $A_{i,j} \neq \perp$, $\mu(B) = 1$, and $\rank(B) \leq d_1$.
  \item $\NO$:  Every positive semi-definite matrix $B \in \R^{n \times n}$, such that $|A_{i,j} - B_{i,j}| \leq \eps$ for all $i,j \in [n]$ with $A_{i,j} \neq \perp$, satisfies $\rank(B) \geq d_2$.
\end{itemize}
\end{definition}
\noindent
Note that the definition restricts the magnitudes of the values in the input partial matrix to at most $1$. Such a restriction is essential when allowing an additive error in the completed matrix, as rank is invariant under scaling.

We first point out that, under plausible complexity assumptions related to the Unique Games Conjecture~\cite{DinurMR06,DinurS10}, the $(d_1,d_2,\eps)$-PSD-Completion problem is intractable for all positive integers $d_2 > d_1 \geq 3$ and real numbers $\eps \in [0,\frac{1}{2})$. Our primary contribution lies in establishing hardness results based solely on the more standard assumption $\P \neq \NP$, as stated below.

\begin{theorem}[Simplified]\label{thm:IntroPSD}
For every sufficiently large integer $d$ and any real number $\eps \in [2^{-O(d)},\frac{1}{7}]$, the $(d,O(\frac{d}{\eps^2 \cdot \log (1/\eps)}),\eps)$-PSD-Completion problem is $\NP$-hard.
\end{theorem}

For admissible values of $d$ and $\eps$, Theorem~\ref{thm:IntroPSD} implies that given a partial matrix $A$ with exposed values of magnitude at most $1$, which admits a positive semi-definite completion with coherence $1$ and rank $d$, it is $\NP$-hard to find a positive semi-definite matrix that agrees with each given value of $A$ up to an additive error of at most $\eps$, even when the rank is allowed to exceed $d$ by a multiplicative factor of $O(\frac{1}{\eps^2 \cdot \log(1/\eps)})$. The theorem encompasses various parameter settings of interest. First, for any fixed approximation factor $\alpha > 1$, there exists some constant $\eps >0$, for which the $(d,\alpha \cdot d,\eps)$-PSD-Completion problem is $\NP$-hard for any sufficiently large integer $d$. Next, letting $\eps$ decrease polynomially with $d$ results in a hardness of approximation factor that is polynomial in $d$. Finally, setting $\eps = 2^{-\Theta(d)}$ yields a hardness of approximation to within a factor of the form $2^{\Omega(d)}$. In fact, for an $\eps$ that decays sufficiently rapidly with $d$, we obtain the following refined hardness result.

\begin{theorem}[Simplified]\label{thm:IntroPSDexp}
For every sufficiently large integer $d$ and any real number $\eps \in [0, 2^{-\Omega(d)}]$, the $(d,2^{(1-o(1)) \cdot d/2},\eps)$-PSD-Completion problem is $\NP$-hard.
\end{theorem}

Theorems~\ref{thm:IntroPSD} and~\ref{thm:IntroPSDexp} substantially strengthen the previously known $\NP$-hardness results for low-rank matrix completion in the positive semi-definite setting. As demonstrated above, our results offer flexibility in the choice of parameters, enabling us to establish hardness for several scenarios of interest. For comparison, the result of~\cite{HMRW14} is specific to $\eps$ decaying polynomially in the rank and achieves $\NP$-hardness of approximation to within factors smaller than $2$. In contrast, for this regime of $\eps$, Theorem~\ref{thm:IntroPSD} establishes $\NP$-hardness of approximation to within factors that grow polynomially in the rank. Furthermore, all our hardness results hold even when the completed matrices of $\YES$ instances are required to have coherence $1$, a property not guaranteed by the $\NP$-hardness result of~\cite{HMRW14}. Finally, while the hardness result of~\cite{ChawinH23} achieves the same gap as Theorem~\ref{thm:IntroPSDexp}, it is restricted to the non-error setting of $\eps=0$.

We turn to our $\NP$-hardness results for the low-rank matrix completion problem, where the completed matrix is constrained by a bounded infinity norm.
Consider the gap $(d_1,d_2,\eps,\theta)$-Completion problem, defined as follows.

\begin{definition}[The $(d_1,d_2,\eps,\theta)$-Completion Problem]\label{def:theta_prob}
For positive integers $d_1 < d_2$ and real numbers $\eps \geq 0$ and $\theta \geq 1$, the {\em $(d_1,d_2,\eps,\theta)$-Completion} problem asks, given a partial matrix
\[A \in ([-\theta,+\theta] \cup \{\perp\})^{n \times n},\]
to distinguish between the following cases.
\begin{itemize}
  \item $\YES$: There exists a matrix $B \in [-\theta,+\theta]^{n \times n}$, such that $A_{i,j} = B_{i,j}$ for all $i,j \in [n]$ with $A_{i,j} \neq \perp$, $\mu(B)=1$, and $\rank(B) \leq d_1$.
  \item $\NO$:  Every matrix $B \in [-\theta,+\theta]^{n \times n}$, such that $|A_{i,j} - B_{i,j}| \leq \eps$ for all $i,j \in [n]$ with $A_{i,j} \neq \perp$, satisfies $\rank(B) \geq d_2$.
\end{itemize}
\end{definition}

Our hardness results for this problem are stated as follows.

\begin{theorem}[Simplified]\label{thm:IntroFit}
For every sufficiently large integer $d$ and any real numbers $\eps \in [2^{-O(d)},\frac{1}{7}]$ and $\theta \in  [1,2^{2^{O(d)}}]$, the $(d,O(\frac{d}{\eps^2 \cdot \log (1/\eps)}),\eps,\theta)$-Completion problem is $\NP$-hard.
\end{theorem}

\begin{theorem}[Simplified]\label{thm:IntroFitExp}
For every sufficiently large integer $d$ and any real numbers $\eps \in [0, 2^{-\Omega(d)}]$ and $\theta \in [1,2^{2^{o(d)}}]$, the $(d,2^{(1-o(1)) \cdot d/2},\eps,\theta)$-Completion problem is $\NP$-hard.
\end{theorem}

\noindent
As mentioned earlier, the previously known hardness results for the $(d_1,d_2,\eps,\theta)$-Completion problem, given in~\cite{HMRW14}, rely on complexity assumptions stronger than the standard conjecture $\P \neq \NP$.

Our $\NP$-hardness results are obtained via an efficient reduction from a gap coloring problem, whose hardness was proved by Krokhin, Opr{\v s}al, Wrochna, and Zivn{\'{y}}~\cite{KrokhinOWZ23}. Their proof employed the concept of line digraphs, introduced by Harary and Norman~\cite{HararyN60} in 1960, which lies at the heart of the present paper as well. Specifically, we introduce extensions of the notions of orthogonality dimension and minrank of graphs (see Definitions~\ref{def:od} and~\ref{def:fit}), and as our main technical contribution, we show that for line digraphs, these quantities are intimately related to the chromatic number. The analysis involves bounds on the rank of perturbed identity matrices that were proved by Alon~\cite{Alon03} in 2003 and have found diverse applications (see Theorem~\ref{thm:Alon} and~\cite{Alon09}). After establishing hardness results for our extensions of orthogonality dimension and minrank, we derive our hardness results for low-rank matrix completion problems. We believe that these novel graph quantities may be of independent interest, and we encourage their further study by pointing out a close relation to the notion of circular chromatic number, introduced by Vince~\cite{Vince88} in 1988 (see Section~\ref{sec:circular}).

\subsection{Proof Technique}

We provide here an overview of the techniques and ideas behind the proofs of our $\NP$-hardness results for the low-rank matrix completion problem. For concreteness, we focus on the PSD-Completion problem, where the completed matrix is required to be positive semi-definite. We then briefly discuss the setting of a bounded infinity norm.

Our hardness proofs are anchored in the classic gap coloring problem. Recall that the chromatic number of a graph $G$, denoted $\chi(G)$, is the smallest number of colors needed for a vertex coloring of $G$ in which adjacent vertices receive distinct colors. For fixed positive integers $k_1<k_2$, the $(k_1,k_2)$-Coloring problem asks, given an input graph $G$, to distinguish between the case where $\chi(G) \leq k_1$ and the case where $\chi(G) \geq k_2$. This problem is known to be intractable for all integers $k_2 > k_1 \geq 3$, under complexity assumptions related to the Unique Games Conjecture~\cite{DinurMR06,DinurS10}. However, identifying the values of $k_1$ and $k_2$ for which the problem is $\NP$-hard has long been considered notoriously difficult. The current state of the art shows that for every integer $k_1 \geq 3$, the problem is $\NP$-hard for $k_2 = 2k_1$, as proved by Barto, Bul{\'{\i}}n, Krokhin, and Opr{\v s}al~\cite{BartoBKO21}, and for $k_2 = \binom{k_1}{\lfloor k_1/2 \rfloor} = 2^{(1-o(1)) \cdot k_1}$, as proved by Krokhin, Opr{\v s}al, Wrochna, and Zivn{\'{y}}~\cite{KrokhinOWZ23}. The latter result, which improves upon the former for all integers $k_1 \geq 6$, serves as the starting point of our hardness proofs.

The $\NP$-hardness proof of Krokhin et al.~\cite{KrokhinOWZ23} for the gap coloring problem is based on the concept of line digraphs~\cite{HararyN60}, which allows to efficiently shrink the chromatic number of a given graph in a controlled manner. Specifically, for a graph $G$, let $\tilde{\delta}G$ denote the underlying graph of the line digraph of $G$, namely, the graph whose vertices are all the ordered pairs of adjacent vertices in $G$ (whose number is twice the number of edges in $G$), where two vertices $(u_1,u_2)$ and $(v_1,v_2)$ are adjacent in $\tilde{\delta}G$ if $u_2=v_1$ or $u_1=v_2$ (see Definition~\ref{def:line}). A result of Poljak and R{\"{o}}dl~\cite{PoljakR81} shows that the chromatic number of the graph $\tilde{\delta}G$ is logarithmic in that of $G$ (see Theorem~\ref{thm:chi_delta}). The hardness result of~\cite{KrokhinOWZ23} was derived by repeatedly applying this transformation to instances of the $(k,2^{O(k^{1/3})})$-Coloring problem, the $\NP$-hardness of which was previously proved by Huang~\cite{Huang13}. Since the appearance of~\cite{KrokhinOWZ23}, line digraphs have been effectively utilized in several hardness proofs, e.g., in~\cite{GuruswamiS20,ChawinH23,HechtMS23}, and they also form a key ingredient in the present paper.

To give a glimpse of the relationship between the chromatic numbers of a graph $G$ and its associated graph $\tilde{\delta}G$, and as a gentle warm-up to our actual argument, let us briefly explain why $\chi(\tilde{\delta}G) \leq k$ implies that $\chi(G) \leq 2^k$. Indeed, given a proper $k$-coloring of $\tilde{\delta}G$, we define a coloring of $G$ by assigning to each vertex $v$ the set $c(v)$ of colors used at the vertices of the form $(\cdot,v)$ in $\tilde{\delta}G$ (i.e., vertices whose head is $v$). Clearly, the number of colors used does not exceed $2^k$. To verify that the coloring is proper, observe that if $u$ and $v$ are adjacent vertices in $G$, the color of the vertex $(u,v)$ in $\tilde{\delta}G$ lies in $c(v)$ but not in $c(u)$, ensuring that $c(u) \neq c(v)$. A slightly better upper bound on the chromatic number of $G$, along with a matching lower bound, is provided in~\cite{PoljakR81}.

A natural extension of graph colorings, originally proposed by Lov{\'a}sz~\cite{Lovasz79} in the introduction of his renowned $\vartheta$-function, is that of orthonormal representations of graphs. A $d$-dimensional orthonormal representation of a graph is an assignment of a unit vector in $\R^d$ to each vertex, such that the vectors assigned to adjacent vertices are orthogonal.\footnote{Strictly speaking, the definition in~\cite{Lovasz79} requires vectors assigned to {\em non-adjacent} vertices to be orthogonal. This corresponds to an orthonormal representation of the complement graph in our terminology, which we adopt from, e.g.,~\cite{Peeters96,BrietBLPS15} for convenience.} The orthogonality dimension of a graph $G$, denoted $\od(G)$, is the smallest integer $d$ for which $G$ admits a $d$-dimensional orthonormal representation (see Definition~\ref{def:od}). Note that for every graph $G$, it holds that
\begin{eqnarray}\label{eq:ODvsChi}
\log_3 \chi(G) \leq \od(G) \leq \chi(G).
\end{eqnarray}
Indeed, for the upper bound on $\od(G)$, notice that a proper $k$-coloring of $G$ yields a $k$-dimensional orthonormal representation by assigning the $i$th vector of the standard basis of $\R^k$ to the vertices of the $i$th color class.
For the lower bound, consider a $k$-dimensional orthonormal representation of $G$, and observe that replacing its vectors by their sign vectors in $\{0,\pm 1\}^k$ results in a proper coloring of $G$ with $3^k$ colors. For a construction of graphs for which the left-hand side of~\eqref{eq:ODvsChi} is tight up to a multiplicative constant, see, e.g.,~\cite{HavivMFCS19}. Now, by associating with each orthonormal representation of $G$ the Gram matrix of its vectors, it follows that $\od(G)$ is the smallest possible rank of a positive semi-definite matrix with ones on the diagonal and zeros in entries that correspond to pairs of adjacent vertices. This formulation naturally connects the problem of determining the orthogonality dimension of a graph to that of determining the minimum rank of a positive semi-definite completion of a suitably defined partial matrix.

The computational hardness of determining the orthogonality dimension of graphs was speculated by Lov{\'a}sz, Saks, and Schrijver~\cite{LovaszSS89} in 1989 and has been studied in several recent works, e.g.,~\cite{BrietBLPS15,GolovnevH22,ChawinH23}.
We first mention that one may combine the inequalities in~\eqref{eq:ODvsChi} with the hardness results of the gap coloring problem from~\cite{DinurMR06,DinurS10} to conclude that deciding whether an input graph $G$ satisfies $\od(G) \leq d_1$ or $\od(G) \geq d_2$ is intractable for all integers $d_2 > d_1 \geq 3$, assuming some variant of the Unique Games Conjecture. This reasoning, however, is insufficient to derive any $\NP$-hardness result for orthogonality dimension, even from the strongest known $\NP$-hardness of gap coloring. Still, it was shown in~\cite{ChawinH23} that for every sufficiently large integer $d$, it is $\NP$-hard to decide whether an input graph $G$ satisfies $\od(G) \leq d$ or $\od(G) \geq 2^{(1-o(1)) \cdot d/2}$. This result was proved based on the hardness of the $(k,2^{(1-o(1))\cdot k})$-Coloring problem from~\cite{KrokhinOWZ23} through the reduction that maps a given graph $G$ to the graph $\tilde{\delta}G$. On the one hand, it follows from~\eqref{eq:ODvsChi} that the orthogonality dimension of $\tilde{\delta}G$ does not exceed its chromatic number, which is logarithmic in $\chi(G)$. To complete the correctness of the reduction, it was shown in~\cite{ChawinH23} that $\od(\tilde{\delta}G) \leq d$ implies that $\chi(G) \leq d^{O(d^2)}$. This, in turn, leads to the $\NP$-hardness of the $d$ vs. $2^{(1-o(1))\cdot d/2}$ gap for the orthogonality dimension problem, as well as for the PSD-Completion problem when no error is allowed.

The argument in~\cite{ChawinH23} was based on the idea outlined below. Suppose that $\od(\tilde{\delta}G) \leq d$, and consider a $d$-dimensional orthonormal representation of $\tilde{\delta}G$, which assigns to each ordered pair $e$ of adjacent vertices in $G$ a vector $x_e \in \R^d$. Associate with each vertex $v$ of $G$ the linear subspace $c(v) \subseteq \R^d$ spanned by the vectors of the form $x_{(\cdot,v)}$. Observe that if $u$ and $v$ are adjacent vertices in $G$, their subspaces $c(u)$ and $c(v)$ are quite distant from each other, in the sense that there exists a vector --- the vector $x_{(u,v)}$ associated with the vertex $(u,v)$ in $\tilde{\delta}G$ --- that lies in $c(v)$ but is orthogonal to the entire subspace $c(u)$. Now, every subspace of $\R^d$ can be represented by an orthonormal basis of at most $d$ vectors. To obtain a coloring of $G$ with finitely many colors, the basis vectors are replaced by their representatives from a sufficiently dense net of the unit sphere in $\R^d$, resulting in a proper coloring of $G$ with $d^{O(d^2)}$ colors, as desired.

In order to adapt this approach to the more tolerant setting of PSD-Completion, where an additive error of $\eps$ is allowed at each entry of the input partial matrix, we introduce the notion of $d$-dimensional $\eps$-orthonormal representations of graphs. Here, each vertex of the graph $G$ at hand is assigned a unit vector in $\R^d$, such that the inner product of vectors assigned to adjacent vertices is at most $\eps$ in absolute value. Letting $\od_\eps(G)$ denote the smallest possible dimension of such an assignment, one can show that approximating the value of $\od_\eps(G)$ for an input graph $G$ is efficiently reducible to approximating the smallest rank of a positive semi-definite matrix that agrees with a given partial matrix, even when an additive error of $O(\eps)$ is allowed per entry. To prove the $\NP$-hardness of the former, we apply again the reduction that transforms a graph $G$ into the graph $\tilde{\delta}G$. For correctness, we aim to establish an upper bound on $\chi(G)$ in terms of $\od_\eps(\tilde{\delta}G)$. It turns out, though, that the proof idea of~\cite{ChawinH23}, described above, does not extend to this setting. To see why, consider a $d$-dimensional $\eps$-orthonormal representation of $\tilde{\delta}G$, which assigns to each ordered pair $e$ of adjacent vertices in $G$ a vector $x_e \in \R^d$, and as before, associate with each vertex $v$ of $G$ the linear subspace $c(v) \subseteq \R^d$ spanned by the vectors of the form $x_{(\cdot,v)}$. Now, let $u$ and $v$ be adjacent vertices in $G$. While the vector $x_{(u,v)}$ still lies in $c(v)$, it is no longer guaranteed to be orthogonal to the subspace $c(u)$. In fact, this vector is only guaranteed to have an inner product of at most $\eps$ in absolute value with the vectors of a basis of $c(u)$, which does not even preclude the possibility that it lies in $c(u)$, making it impossible to deduce that $c(u)$ and $c(v)$ are far apart. As an example, consider the two unit vectors $e_1$ and $\sqrt{1-\eps^2} \cdot e_1 + \eps \cdot e_2$, where $e_i$ denotes the $i$th vector of the standard basis of $\R^d$. Notice that the vector $e_2$ has an inner product of at most $\eps$ in absolute value with each of the two vectors, yet it lies within the subspace they span.

We overcome the above difficulty through a different coloring of $G$. Namely, for each vertex $v$ of $G$, we consider a maximal set $c(v) \subseteq \R^d$ of vectors of the form $x_{(\cdot,v)}$ with pairwise inner products of absolute value at most $\eps'$, for some appropriately chosen $\eps' > \eps$. Now, if $u$ and $v$ are adjacent vertices in $G$, then the vector $x_{(u,v)}$ has an inner product of at most $\eps$ in absolute value with each vector of $c(u)$. However, by the maximality of $c(v)$, there must exist a vector in $c(v)$ whose inner product with $x_{(u,v)}$ is larger than $\eps'$ in absolute value: either $x_{(u,v)}$ itself or some other vector that prevented it from being added to $c(v)$. This, in a sense, implies that the sets $c(u)$ and $c(v)$ are sufficiently far apart, so that they remain distinct once their vectors are replaced by representatives from a sufficiently dense net. Yet, the number of vectors in the sets $c(v)$ has a significant effect on the number of colors used. In contrast to the $\eps=0$ setting, here we do not handle bases of subspaces of $\R^d$, and therefore, we cannot ensure that the size of each set $c(v)$ is bounded by $d$. We do know that the vectors of each set $c(v)$ are nearly orthogonal to each other, with the absolute value of their pairwise inner products at most $\eps'$. To bound their size, we apply bounds of Alon~\cite{Alon03} on the rank of perturbed identity matrices (see Theorem~\ref{thm:Alon}). When $\eps'$ is sufficiently small, namely, $\eps' \leq O(1/\sqrt{d})$, it turns out that each set $c(v)$ includes fewer than $2d$ vectors. For larger values of $\eps'$, the bound weakens, leading to the dependence of the hardness gap on the error $\eps$, as stated in Theorem~\ref{thm:IntroPSD}.

To extend our approach to the low-rank matrix completion problem with a bounded infinity norm, we introduce an extension of the notion of graph-fitting matrices, proposed by Haemers in~\cite{Haemers81}. Here, for a graph $G$, we consider matrices (not necessarily positive semi-definite) that have ones on the diagonal and values of magnitude at most $\eps$ in entries corresponding to adjacent vertices (see Definition~\ref{def:fit}). Our main technical result provides an upper bound on the chromatic number of a graph $G$, assuming that the graph $\tilde{\delta}G$ admits such a matrix with a bounded rank and a bounded infinity norm (see Theorem~\ref{thm:chi(H)} and Corollary~\ref{cor:fit_H}). The proof generalizes the technique described above, drawing on the fact that every matrix has a rank factorization involving matrices with rows of bounded norm (see Lemma~\ref{lemma:factor}). The detailed argument is presented in the subsequent technical sections.

\subsection{Outline}
The rest of the paper is organized as follows. In Section~\ref{sec:prel}, we collect several notations and tools that will be used throughout the paper. In Section~\ref{sec:OD}, we define and study nearly orthonormal representations of graphs, with particular attention to line digraphs. In Section~\ref{sec:hard}, we apply our insights to establish the hardness of a problem we call Graph Fitness. This, in turn, leads to the hardness of the $(d_1,d_2,\eps)$-PSD-Completion and $(d_1,d_2,\eps,\theta)$-Completion problems, thereby verifying Theorems~\ref{thm:IntroPSD},~\ref{thm:IntroPSDexp},~\ref{thm:IntroFit}, and~\ref{thm:IntroFitExp} (for precise statements, see Corollaries~\ref{cor:PSD_hard},~\ref{cor:Comp_hard}, and~\ref{cor:PSD_cases}).

\section{Preliminaries}\label{sec:prel}

Throughout the paper, we omit floor and ceiling signs when they are not essential, and all logarithms are taken in base $2$, unless otherwise specified. Graphs refer to undirected graphs, and digraphs refer to directed graphs, with all graphs and digraphs being simple (i.e., with no loops or parallel edges). For a positive integer $n$, we denote $[n] = \{1,\ldots,n\}$.

\subsection{Linear Algebra}
For a positive integer $d$, let $\langle \cdot,\cdot \rangle$ and $\|\cdot\|$ stand for the standard inner product and Euclidean norm on $\R^d$, respectively.
As is customary, a vector $x \in \R^d$ with $\|x\|=1$ is referred to as a unit vector.
The following simple claim will be used.
\begin{claim}\label{claim:x,y,z}
For every positive integer $d$ and any three vectors $x,y,z \in \R^d$, it holds that
\[ \Big | |\langle x,y \rangle| - |\langle z,y \rangle | \Big | \leq \|x-z\| \cdot \| y\|. \]
\end{claim}

\begin{proof}
Observe that
\[ \Big | |\langle x,y \rangle| - |\langle z,y \rangle | \Big | \leq \Big | \langle x,y \rangle - \langle z,y \rangle  \Big | = | \langle x-z,y \rangle | \leq \|x-z\| \cdot \| y\|,\]
where the first inequality relies on the triangle inequality, and the second on the Cauchy--Schwarz inequality. This completes the proof.
\end{proof}

For a real matrix $A = (a_{i,j})$, let $\rank(A)$ denote its rank over $\R$, and let $\|A\|_\infty$ denote its infinity norm, defined by $\|A\|_\infty = \max_{i,j}|a_{i,j}|$.
It is well known that every matrix $A \in \R^{n \times n}$ of rank $d$ can be expressed as  $A = X \cdot Y^t$ for some matrices $X,Y \in \R^{n \times d}$. The following lemma guarantees the existence of such a factorization with matrices $X$ and $Y$ whose rows have bounded norms. Its proof relies on John's classical theorem from Banach space theory and can be found, e.g., in~\cite[Corollary~2.2]{Rashtchian16}.

\begin{lemma}\label{lemma:factor}
Let $d \leq n$ be positive integers. For every matrix $A \in \R^{n \times n}$ of rank $d$, there exist two matrices $X,Y \in \R^{n \times d}$ satisfying $A = X \cdot Y^t$, such that every row of $X$ and $Y$ has norm at most $d^{1/4} \cdot \|A\|_\infty^{1/2}$.
\end{lemma}

A matrix $A \in \R^{n \times n}$ is said to be {\em positive semi-definite} if $x^t A x \geq 0$ for all vectors $x \in \R^n$. This condition is equivalent to the existence of a matrix $X\in \R^{n \times d}$ such that $A = X \cdot X^t$ where $d = \rank(A)$. The coherence of a symmetric matrix measures the extent to which its row (or column) space aligns with the vectors of the standard basis.
It is formally defined as follows.
\begin{definition}[Coherence]\label{def:coherence}
For positive integers $d \leq n$, let $U$ be a $d$-dimensional subspace of $\R^n$, and let $P_U$ be the orthogonal projection onto $U$.
The {\em coherence} of $U$ is defined as $\mu(U) = \frac{n}{d} \cdot \max_{i \in [n]}\|P_U e_i \|^2$, where $e_i$ stands for the $i$th vector of the standard basis of $\R^n$.
Note that $\mu(U) \in [1,\frac{n}{d}]$.
The {\em coherence} of a symmetric matrix $A \in \R^{n\times n}$, denoted by $\mu(A)$, is defined as the coherence of its row (or column) space.
\end{definition}

\subsection{Nets}

For a positive integer $d$ and a real number $\theta >0$, let $B_d(\theta)$ denote the closed $d$-dimensional ball of radius $\theta$ centered at the origin, i.e., $B_d(\theta) = \{ x \in \R^d~\mid~\|x\| \leq \theta\}$. We define a net for a closed ball as follows.
\begin{definition}\label{def:net}
For a positive integer $d$ and real numbers $\eta,\theta >0$, an {\em $\eta$-net} for $B_d(\theta)$ is a set $K \subseteq \R^d$ such that for any $x \in B_d(\theta)$, there exists a point $y \in K$ satisfying $\|x-y\| < \eta$.
\end{definition}

Note that Definition~\ref{def:net} requires every point in the ball to admit a point in the $\eta$-net at a distance {\em strictly} smaller than $\eta$. We need the following standard lemma on the existence of bounded-size nets for balls (see, e.g.,~\cite{FLM77}). A brief proof is included for completeness.
\begin{lemma}\label{lemma:net}
For every positive integer $d$ and any real numbers $\eta,\theta >0$, there exists an $\eta$-net for $B_d(\theta)$ of size at most \[ \bigg (\frac{2\theta}{\eta}+1 \bigg )^d.\]
\end{lemma}

\begin{proof}
Let $K$ be a maximal subset of $B_d(\theta)$ with pairwise distances at least $\eta$. By maximality, $K$ forms an $\eta$-net for $B_d(\theta)$. By the triangle inequality, the open balls of radius $\frac{\eta}{2}$ centered at the points of $K$ are pairwise disjoint and contained in $B_d(\theta+\frac{\eta}{2})$. Therefore, the number of points in $K$ does not exceed the ratio between the volumes, implying that
\[ |K| \leq \bigg ( \frac{\theta+\eta/2}{\eta/2} \bigg )^d = \bigg ( \frac{2\theta}{\eta} +1 \bigg )^d.\]
The proof is complete.
\end{proof}

\subsection{The Rank of Perturbed Identity Matrices}

It is well known and easy to verify that if a matrix $A=(a_{i,j}) \in \R^{m \times m}$ satisfies $a_{i,i}=1$ for all $i \in [m]$ and $|a_{i,j}| \leq \frac{1}{m}$ for all distinct $i,j \in [m]$, then $A$ has full rank. The following theorem, proved by Alon~\cite{Alon03}, provides lower bounds on the rank of a symmetric matrix under a weaker assumption on its off-diagonal entries. For a variety of applications of these bounds, the reader is referred to~\cite{Alon09} (see also~\cite{AKMMR06}).

\begin{theorem}[\cite{Alon03}]\label{thm:Alon}
There exists an absolute constant $c>0$ for which the following holds.
For positive integers $d \leq m$ and for a real number $\eps \in [0,1)$, let $A = (a_{i,j}) \in \R^{m \times m}$ be a symmetric matrix of rank $d$ satisfying $a_{i,i}=1$ for all $i \in [m]$ and $|a_{i,j}| \leq \eps$ for all distinct $i,j \in [m]$. Then
\begin{enumerate}
  \item\label{ThmAlonItm:1} $d \geq \frac{m}{1+\eps^2 \cdot (m-1)}$, and
  \item\label{ThmAlonItm:2} if $\eps \in [\frac{1}{\sqrt{m}},\frac{1}{2}]$, then $d \geq c \cdot \frac{\log m}{\eps^2 \cdot \log (1/\eps)}$.
\end{enumerate}
\end{theorem}

In light of Theorem~\ref{thm:Alon}, we introduce the quantities $m(d,\eps)$, defined as follows.

\begin{definition}\label{def:m(r,eps)}
For a positive integer $d$ and a real number $\eps \in [0,\frac{1}{2}]$, let $m(d,\eps)$ denote the maximum integer $m$ for which there exists a symmetric matrix $A = (a_{i,j}) \in \R^{m \times m}$ of rank (at most) $d$ satisfying $a_{i,i}=1$ for all $i \in [m]$ and $|a_{i,j}| \leq \eps$ for all distinct $i,j \in [m]$.
\end{definition}
\noindent
Note that for all positive integers $d$ and real numbers $\eps \leq \eps'$, it holds that $m(d,\eps) \leq m(d,\eps')$.

As a direct corollary of Theorem~\ref{thm:Alon}, we obtain the following bounds on $m(d,\eps)$.

\begin{corollary}\label{cor:m(d,eps)}
There exists an absolute constant $c$ such that the following holds for all positive integers $d$.
\begin{enumerate}
  \item\label{CorAlonItm:1} If $\eps \in  [0,\frac{1}{\sqrt{d}})$, then $m(d,\eps) \leq d \cdot \frac{1-\eps^2}{1-d \cdot \eps^2}$.
  In particular, if $\eps \leq \frac{1}{\sqrt{2d}}$, then $m(d,\eps)<2d$.
  \item\label{CorAlonItm:2} If $\eps \in [\frac{1}{\sqrt{d}},\frac{1}{2}]$, then $m(d,\eps) \leq 2^{c \cdot d \cdot \eps^2 \cdot \log(1/\eps)}$.
\end{enumerate}
\end{corollary}

\begin{proof}
Let $A = (a_{i,j}) \in \R^{m \times m}$ be a symmetric matrix of rank $d$ satisfying $a_{i,i}=1$ for all $i \in [m]$ and $|a_{i,j}| \leq \eps$ for all distinct $i,j \in [m]$.
First, by Item~\ref{ThmAlonItm:1} of Theorem~\ref{thm:Alon}, it holds that $d \geq \frac{m}{1+\eps^2 \cdot (m-1)}$.
For any $\eps < \frac{1}{\sqrt{d}}$, rearranging the terms yields that $m \leq d \cdot \frac{1-\eps^2}{1-d \cdot \eps^2}$, implying the bound on $m(d,\eps)$ stated in Item~\ref{CorAlonItm:1} of the corollary.
Next, consider some $\eps \in [\frac{1}{\sqrt{d}},\frac{1}{2}]$. By $m \geq d$, we can apply Item~\ref{ThmAlonItm:2} of Theorem~\ref{thm:Alon} to obtain that, for an absolute constant $c>0$, it holds that $d \geq c \cdot \frac{\log m}{\eps^2 \cdot \log (1/\eps)}$, which by rearranging yields that $m \leq 2^{d \cdot \eps^2 \cdot \log(1/\eps)/c}$. This implies the bound on $m(d,\eps)$ stated in Item~\ref{CorAlonItm:2} of the corollary for an appropriate absolute constant.
\end{proof}

\section{Nearly Orthonormal Representations of Graphs}\label{sec:OD}

A $d$-dimensional orthonormal representation of a graph is an assignment of a unit vector in $\R^d$ to each vertex, such that adjacent vertices receive orthogonal vectors (see~\cite{Lovasz79}). We introduce the following relaxation of this concept, where adjacent vertices receive vectors that are {\em nearly} orthogonal.

\begin{definition}\label{def:od}
Let $G=(V,E)$ be a graph.
For a positive integer $d$ and a real number $\eps \in [0,1)$, a {\em $d$-dimensional $\eps$-orthonormal representation} of $G$ is an assignment of a unit vector $x_v \in \R^d$ to each vertex $v \in V$, such that for every pair of adjacent vertices $u$ and $v$ in $G$, it holds that $|\langle x_u,x_v \rangle| \leq \eps$.
For any real number $\eps \in [0,1)$, the {\em $\eps$-orthogonality dimension} of $G$, denoted $\od_\eps(G)$, is the smallest positive integer $d$ for which $G$ admits a $d$-dimensional $\eps$-orthonormal representation. We omit $\eps$ from the notation and terminology when $\eps=0$.
\end{definition}

A well-studied extension of orthonormal representations, introduced in~\cite{Haemers81}, is that of graph-fitting matrices (see also~\cite{Peeters96}). We propose the following relaxation of this notion.

\begin{definition}\label{def:fit}
Let $G=(V,E)$ be a graph. For a real number $\eps \in [0,1)$, a matrix $A=(a_{u,v}) \in \R^{|V| \times |V|}$, whose rows and columns are indexed by $V$, is said to {\em $\eps$-fit} the graph $G$ if $a_{v,v}=1$ for all $v \in V$ and $|a_{u,v}| \leq \eps$ whenever $u$ and $v$ are adjacent vertices in $G$. When $\eps=0$, $A$ is said to {\em fit} $G$.
\end{definition}
\noindent
For a given graph $G$ and a real number $\eps \in [0,1)$, we are concerned with the minimum possible rank of a matrix that $\eps$-fits $G$.
When $\eps=0$, this quantity coincides with the minrank of the complement of $G$ over the reals (see~\cite{Haemers81,Peeters96}).

\begin{remark}\label{remark:ODvsFIT}
The notions of $\eps$-orthonormal representations and $\eps$-fitting matrices, given in Definitions~\ref{def:od} and~\ref{def:fit}, are closely related.
To see this, consider a graph $G=(V,E)$, and associate with each $d$-dimensional $\eps$-orthonormal representation $(x_v)_{v \in V}$ of $G$ the Gram matrix $A =(a_{u,v}) \in \R^{|V| \times |V|}$ of its vectors, defined by $a_{u,v} = \langle x_u,x_v \rangle$ for all $u,v \in V$. Note that such a matrix $A$ is positive semi-definite, has rank at most $d$, and $\eps$-fits the graph $G$. Therefore, $d$-dimensional $\eps$-orthonormal representations of a graph $G$ may be regarded as the special case of matrices of rank at most $d$ that $\eps$-fit $G$, with the additional property of positive semi-definiteness.
\end{remark}

In the rest of this section, we relate the chromatic number of a graph to the rank of matrices that nearly fit it. We first establish such relations for general graphs and then proceed to the case of underlying graphs of line digraphs. Finally, we link the notion of nearly orthonormal representations to the circular chromatic number.

\subsection{Chromatic Number}

For a positive integer $k$, a {\em $k$-coloring} of a graph $G$ is a mapping from the vertex set of $G$ to a set of size $k$. A coloring $c$ of $G$ is called {\em proper} if $c(u) \neq c(v)$ whenever $u$ and $v$ are adjacent vertices in $G$. The graph $G$ is called {\em $k$-colorable} if it admits a proper $k$-coloring, and the smallest integer $k$ for which $G$ is $k$-colorable is called the {\em chromatic number} of $G$ and is denoted by $\chi(G)$.
Observe that any proper $k$-coloring of $G$ induces a $k$-dimensional orthonormal representation of $G$, in which the vertices of the $i$th color class are assigned the $i$th vector of the standard basis of $\R^k$. Consequently, every graph $G$ satisfies $\od(G) \leq \chi(G)$ and thus admits a positive semi-definite matrix of rank at most $\chi(G)$ that fits it (see Remark~\ref{remark:ODvsFIT}). The following simple lemma, whose argument is borrowed from~\cite{HMRW14}, shows that a slight modification of $G$ ensures the existence of such a matrix with minimal coherence (recall Definition~\ref{def:coherence}).

\begin{lemma}\label{lemma:copies}
For positive integers $k$ and $n$, let $G$ be a $k$-colorable graph on $n$ vertices, and let $H$ denote the disjoint union of $k$ copies of $G$.
Then there exists a positive semi-definite matrix $B \in \{0,1\}^{kn \times kn}$ that fits the graph $H$, such that $\rank(B) = k$ and $\mu(B)=1$ .
\end{lemma}

\begin{proof}
For a $k$-colorable graph $G=(V,E)$ on $n$ vertices, let $c:V \rightarrow [k]$ be a proper $k$-coloring of $G$.
It is clear that any cyclic shift of the coloring $c$, where each color $i$ is replaced by the color in $[k]$ that is congruent to $i+j$ modulo $k$ for some fixed $j$, results in another proper $k$-coloring of $G$. Let $H=(V_H,E_H)$ denote the disjoint union of $k$ copies of $G$, and let $c'$ denote the $k$-coloring of $H$ that colors each copy of $G$ by a different cyclic shift of $c$. Notice that $c'$ is a proper $k$-coloring of $H$ with each color class of size exactly $n$. For each $i \in [k]$ and each vertex $v$ of $H$ colored $i$ by $c'$, let $x_v$ be the $i$th vector of the standard basis of $\R^k$. The assignment $(x_v)_{v \in V_H}$ forms a $k$-dimensional orthonormal representation of $H$, whose Gram matrix $B = (\langle x_u,x_v \rangle) \in \{0,1\}^{kn \times kn}$ is positive semi-definite, fits $H$, and satisfies $\rank(B) = k$. Additionally, the row space $U$ of $B$ is spanned by $k$ vectors in $\{0,1\}^{k n}$ with pairwise disjoint supports, each of size $n$. By scaling these vectors by a factor of $\frac{1}{\sqrt{n}}$, we obtain an orthonormal basis of $U$. Each standard basis vector $e_i \in \R^{kn}$ has an inner product of $\frac{1}{\sqrt{n}}$ with one basis vector of $U$ and zero with all the others. Therefore, letting $P_U$ denote the orthogonal projection onto $U$, it follows that $\|P_U e_i\|^2 = \frac{1}{n}$. This implies that the coherence of $B$ satisfies $\mu(B)=\frac{kn}{k} \cdot \frac{1}{n} =1$, and we are done.
\end{proof}

The following theorem relates the chromatic number of a graph to the rank of a matrix that nearly fits it. It will be used to obtain conditional hardness results for low-rank matrix completion problems (see the discussion at the end of Section~\ref{sec:hard}). A similar reasoning appears in~\cite{HMRW14}.

\begin{theorem}\label{thm:chi(G)}
Let $G=(V,E)$ be a graph.
For a positive integer $d$ and real numbers $\eps \in [0,1)$ and $\theta \geq 1$, suppose that there exist two matrices $X,Y \in \R^{|V| \times d}$, where each row of $X$ and $Y$ has norm at most $\theta$, such that the matrix $X \cdot Y^t$ $\eps$-fits the graph $G$.
Then, it holds that \[\chi(G) \leq \bigg (\frac{4 \theta^2}{1-\eps}+1 \bigg )^{d}.\]
Furthermore, if $X=Y$, then
\[\chi(G) \leq \bigg (\frac{2 \sqrt{2}}{\sqrt{1-\eps}}+1 \bigg )^{d}.\]
\end{theorem}

\begin{proof}
Consider a graph $G$, an integer $d$, real numbers $\eps,\theta$, and matrices $X,Y$ as in the statement of the theorem.
Set $\eta = \frac{1-\eps}{2\theta}$. By Lemma~\ref{lemma:net}, there exists an $\eta$-net $K$ for the closed $d$-dimensional ball $B_d(\theta)$ of radius $\theta$, such that $|K| \leq (\frac{2\theta}{\eta}+1)^d $. Let $f:B_d(\theta) \rightarrow K$ be a function that maps each point $x \in B_d(\theta)$ to a point in $K$ that is closest to $x$. Since $K$ is an $\eta$-net for $B_d(\theta)$, it holds that $\|f(x)-x\| < \eta$ for every $x \in B_d(\theta)$.

For every vertex $v \in V$ of $G$, let $x_v$ and $y_v$ denote the rows associated with $v$ in the given matrices $X$ and $Y$, respectively. By assumption, it holds that $\|x_v\| \leq \theta$ and $\|y_v\| \leq \theta$ for all $v \in V$.
Consider the coloring of $G$ that assigns to each vertex $v \in V$ the vector $f(x_v)$.
Obviously, the number of used colors is at most $|K| \leq (\frac{2\theta}{\eta}+1)^d = (\frac{4\theta^2}{1-\eps}+1)^d$.
It remains to show that the coloring is proper.

Let $u,v \in V$ be adjacent vertices in $G$. Since $X \cdot Y^t$ $\eps$-fits the graph $G$, it follows that $\langle x_v,y_v \rangle =1$ and $|\langle x_u,y_v \rangle | \leq \eps$.
Using Claim~\ref{claim:x,y,z}, this yields that
\[ |\langle f(x_u), y_v \rangle | \leq |\langle x_u,y_v \rangle | + \| f(x_u) - x_u \| \cdot \| y_v \| < \eps + \eta \cdot \theta,\]
whereas
\[ |\langle f(x_v), y_v \rangle | \geq |\langle x_v,y_v \rangle | - \| f(x_v) - x_v \| \cdot \| y_v \| > 1 - \eta \cdot \theta = \eps + \eta \cdot \theta, \]
where the last equality follows from the definition of $\eta$.
This implies that $f(x_u) \neq f(x_v)$, hence the colors assigned to $u$ and $v$ are distinct, as desired.

Finally, suppose that $X=Y$, and set  $\theta = 1$ and $\eta = \sqrt{\frac{1-\eps}{2}}$.
The number of colors used by the above coloring of $G$ is at most $|K| \leq (\frac{2\theta}{\eta}+1)^d = (\frac{2\sqrt{2}}{\sqrt{1-\eps}}+1)^d$.
We show that this coloring is proper. Let $u,v \in V$ be adjacent vertices in $G$. Since $X \cdot X^t$ $\eps$-fits the graph $G$, it follows that $\|x_u\| = \|x_v\| = 1$ and that $| \langle x_u,x_v \rangle | \leq \eps$, implying that
\begin{eqnarray}\label{eq:x_u-x_v}
\|x_u-x_v\|^2 = 2-2 \cdot \langle x_u,x_v \rangle \geq 2 \cdot (1-\eps).
\end{eqnarray}
This yields that
\begin{eqnarray*}
\|f(x_u) - f(x_v) \| &=& \| x_u - x_v + (f(x_u) - x_u) - (f(x_v)-x_v) \| \\
& \geq & \|x_u-x_v\|-\| f(x_u) - x_u\| - \| f(x_v)-x_v \| \\
& > & \sqrt{2 \cdot (1-\eps)} - 2\eta =0,
\end{eqnarray*}
where the first inequality relies on the triangle inequality, and the second on~\eqref{eq:x_u-x_v}.
It thus follows that $f(x_u) \neq f(x_v)$, so we are done.
\end{proof}

\subsection{Chromatic Number of Line Digraphs}

The concept of line digraphs, introduced in~\cite{HararyN60}, is defined as follows.
\begin{definition}[Line Digraph]\label{def:line}
For a digraph $G = (V,E)$, the {\em line digraph of $G$}, denoted $\delta G$, is the digraph on the vertex set $E$, where there is a directed edge from a vertex $(u_1,u_2)$ to a vertex $(v_1,v_2)$ whenever $u_2=v_1$.
For an (undirected) graph $G$, its line digraph $\delta G$ is defined as the line digraph of the digraph obtained from $G$ by replacing each edge with two oppositely directed edges. Let $\tilde{\delta} G$ denote the underlying graph of $\delta G$, i.e., the graph obtained from $\delta G$ by ignoring the directions of the edges.
\end{definition}

The following result, proved by Poljak and R{\"{o}}dl~\cite{PoljakR81} (see also~\cite{HarnerE72}), shows that the chromatic number of a graph $G$ determines the chromatic number of $\tilde{\delta} G$. The statement involves the function $b: \N \rightarrow \N$, defined by $b(n) = \binom{n}{\lfloor n/2 \rfloor}$.

\begin{theorem}[\cite{PoljakR81}]
\label{thm:chi_delta}
For every graph $G$, $\chi (\tilde{\delta} G ) = \min \{ n \in \N \mid \chi(G) \leq b(n) \}$.
\end{theorem}

The following theorem ties the chromatic number of a graph to the rank of a symmetric matrix that nearly fits the underlying graph of its line digraph. It plays a crucial role in our $\NP$-hardness results. The statement involves the quantities $m(d,\eps)$ given in Definition~\ref{def:m(r,eps)}.

\begin{theorem}\label{thm:chi(H)}
Let $G=(V,E)$ be a graph, and let $\tilde{\delta} G=(V',E')$ be the underlying graph of its line digraph.
For a positive integer $d$ and real numbers $\eps \in [0,\frac{1}{2})$ and $\theta \geq 1$, suppose that there exist two matrices $X,Y \in \R^{|V'| \times d}$, where each row of $X$ and $Y$ has norm at most $\theta$, such that $X \cdot Y^t$ is a symmetric matrix that $\eps$-fits the graph $\tilde{\delta} G$.
Then, for any $\eta \in (0, \frac{1-2\eps}{4\theta}]$, it holds that \[\chi(G) \leq \bigg (\frac{2\theta}{\eta} +1 \bigg )^{d \cdot m(d,2\eta \theta+\eps)}.\]
\end{theorem}

\begin{proof}
Consider a graph $G$, an integer $d$, real numbers $\eps,\theta,\eta$, and matrices $X,Y$ as in the statement of the theorem.
By Lemma~\ref{lemma:net}, there exists an $\eta$-net $K$ for the closed $d$-dimensional ball $B_d(\theta)$ of radius $\theta$, such that $|K| \leq (\frac{2\theta}{\eta}+1)^d$.
Let $f:B_d(\theta) \rightarrow K$ be a function that maps each point $x \in B_d(\theta)$ to a point in $K$ that is closest to $x$. Since $K$ is an $\eta$-net for $B_d(\theta)$, it holds that $\|f(x)-x\| < \eta$ for every $x \in B_d(\theta)$.

Recall that every vertex of $\tilde{\delta} G$ is a pair $e=(u,v) \in V'$ of vertices $u,v \in V$ that are adjacent in $G$.
For each such vertex $e$, let $x_e$ and $y_e$ denote the rows associated with $e$ in the given matrices $X$ and $Y$, respectively.
By assumption, $\|x_e\| \leq \theta$ and $\|y_e\| \leq \theta$ for every $e \in V'$.
Since the matrix $X \cdot Y^t$ $\eps$-fits $\tilde{\delta} G$, it follows that $\langle x_e,y_e \rangle = 1$ for every $e \in V'$, and that $|\langle x_e,y_{e'} \rangle | \leq \eps$ whenever $e$ and $e'$ are adjacent in $\tilde{\delta}G$.

We define a coloring of $G$ as follows.
Set $\eps' = 2\eta \theta+\eps \leq \frac{1}{2}$.
For every vertex $v \in V$, consider the set $E_v$ of vertices of $\tilde{\delta} G$ whose head is $v$, that is,
\[E_v = \{e \in V' \mid e=(u,v) \mbox{~for some~}u \in V\}.\]
Let $E'_v$ be a maximal subset of $E_v$ (with respect to containment), such that for all distinct $e,e' \in E'_v$, it holds that $|\langle x_e,y_{e'} \rangle | \leq \eps'$. Equivalently, we require the sub-matrix of $X \cdot Y^t$, restricted to the rows and columns corresponding to the vertices of $E'_v$, to have off-diagonal values at most $\eps'$ in absolute value. Notice that $X \cdot Y^t$ is a symmetric matrix of rank at most $d$, and thus so is each of its principal sub-matrices.
Letting $m = m(d,\eps')$, it follows that $|E'_v| \leq m$. Now, we assign to each vertex $v \in V$ the color $c(v)$, defined as the set of all vectors $f(x_e)$ with $e \in E'_v$.
The number of colors used by the coloring $c$ does not exceed the number of $m$-tuples of vectors from $K$, which is $|K|^{m} \leq (\frac{2\theta}{\eta}+1)^{d \cdot m}$. It remains to show that the coloring $c$ is proper.

Let $u,v \in V$ be adjacent vertices in $G$, and consider the vector $y_{(u,v)}$ associated with the vertex $(u,v)$ of $\tilde{\delta} G$ in the matrix $Y$.
Since every vertex of $E_u$ is adjacent in $\tilde{\delta} G$ to the vertex $(u,v)$, it follows that every $e \in E_u$ satisfies $|\langle x_e, y_{(u,v)} \rangle | \leq \eps$.
Using Claim~\ref{claim:x,y,z}, this yields that every $e \in E'_u \subseteq E_u$ satisfies
\[ |\langle f(x_e), y_{(u,v)} \rangle | \leq |\langle x_e,y_{(u,v)} \rangle | + \| f(x_e) - x_e \| \cdot \| y_{(u,v)} \| < \eps + \eta \theta.\]
We next argue that there exists a vertex $e \in E'_v$ such that $|\langle x_e,y_{(u,v)} \rangle | > \eps'$.
Indeed, the vertex $(u,v)$ lies in $E_v$. If $(u,v)$ lies in $E'_v$, then we have $|\langle x_{(u,v)},y_{(u,v)} \rangle | =1 > \eps'$, and otherwise, the maximality of $E'_v$ combined with the symmetry of $X \cdot Y^t$ implies the existence of the desired vertex $e$.
Using Claim~\ref{claim:x,y,z} again, it follows that this vertex $e$ satisfies
\[ |\langle f(x_e), y_{(u,v)} \rangle | \geq |\langle x_e,y_{(u,v)} \rangle | - \| f(x_e) - x_e \| \cdot \| y_{(u,v)} \| > \eps' - \eta \theta = \eps + \eta \theta. \]
We conclude that some vector $f(x_e)$ in the set $c(v)$ is different from all the vectors in the set $c(u)$, hence $c(u) \neq c(v)$.
Therefore, the coloring $c$ of $G$ is proper, and we are done.
\end{proof}

As a simple consequence of Theorem~\ref{thm:chi(H)}, we obtain the following result.

\begin{corollary}\label{cor:OD_H}
For every graph $G$, a positive integer $d$, and real numbers $\eps \in [0,\frac{1}{2})$ and $\eta \in (0, \frac{1-2\eps}{4}]$, if $\od_\eps(\tilde{\delta} G) \leq d$ then
\[\chi(G) \leq \bigg (\frac{2}{\eta}+1 \bigg )^{d \cdot m(d,2\eta+\eps)}.\]
\end{corollary}

\begin{proof}
For a graph $G$, let $\tilde{\delta} G = (V',E')$, suppose that $\od_\eps(\tilde{\delta} G) \leq d$, and consider a $d$-dimensional $\eps$-orthonormal representation $(x_e)_{e \in V'}$ of $\tilde{\delta} G$. Let $X$ be the $|V'| \times d$ real matrix, whose rows are indexed by $V'$, where the row associated with a vertex $e \in V'$ is $x_e$. Note that the norm of each row of $X$ is $1$. Observe that the matrix $X \cdot X^t$ is symmetric and $\eps$-fits the graph $\tilde{\delta} G$. The proof is completed by applying Theorem~\ref{thm:chi(H)} with both matrices set to $X$ and with $\theta=1$.
\end{proof}

By combining Theorem~\ref{thm:chi(H)} with Lemma~\ref{lemma:factor}, we derive the following further consequence.

\begin{corollary}\label{cor:fit_H}
For a graph $G=(V,E)$, a positive integer $d$, and a real number $\eps \in [0,\frac{1}{2})$, let $A$ be a real matrix of rank $d$ that $\eps$-fits the graph $\tilde{\delta} G$, and set $\theta = (2d)^{1/4} \cdot \|A\|_\infty^{1/2}$.
Then, for any $\eta \in (0, \frac{1-2\eps}{4\theta}]$, it holds that
\[\chi(G) \leq \bigg (\frac{2\theta}{\eta} +1 \bigg )^{2d \cdot m(2d,2\eta \theta+\eps)}.\]
\end{corollary}

\begin{proof}
For a graph $G$, let $\tilde{\delta} G = (V',E')$, and consider a matrix $A$ as in the statement of the corollary.
Put $B = \tfrac{1}{2} \cdot (A+A^t)$, and observe that $B$ is a symmetric matrix of rank at most $2d$ with $\|B\|_\infty \leq \|A\|_\infty$ that $\eps$-fits the graph $\tilde{\delta} G$.
By Lemma~\ref{lemma:factor}, there exist two matrices $X,Y \in \R^{|V'| \times (2d)}$ satisfying $B = X \cdot Y^t$, such that every row of $X$ and $Y$ has norm at most $(2d)^{1/4} \cdot \|B\|_\infty^{1/2} \leq \theta$.
The proof is completed by applying Theorem~\ref{thm:chi(H)} with the matrices $X,Y$ and with $\theta$.
\end{proof}

\subsection{Circular Chromatic Number}\label{sec:circular}

The circular chromatic number of graphs, introduced by Vince~\cite{Vince88}, has several equivalent definitions, one of which is presented below. For a comprehensive introduction to the topic, the reader is referred to the survey~\cite{ZhuSurvey}.

\begin{definition}
The {\em circular chromatic number} of a graph $G=(V,E)$, denoted $\chi_c(G)$, is the infimum of all real numbers $r \geq 1$ that admit a mapping $f:V \rightarrow [0,1)$, such that for every pair of adjacent vertices $u$ and $v$ in $G$, it holds that $\frac{1}{r} \leq |f(u)-f(v)| \leq 1-\frac{1}{r}$.
\end{definition}

It is known that the infimum in the definition of $\chi_c(G)$ is always attained (at a rational number), hence the infimum can be replaced by a minimum.
It is also known that any graph $G$ satisfies $\chi(G) = \lceil \chi_c(G) \rceil$. The following observation relates the circular chromatic number to $2$-dimensional nearly orthonormal representations. Note that every graph $G$ with at least one edge satisfies $\chi_c(G) \geq 2$.

\begin{proposition}\label{prop:circular}
For every graph $G$ with at least one edge and for any real number $\eps \in [0,1)$, it holds that $\od_\eps(G) \leq 2$ if and only if $\eps \geq \cos(\frac{\pi}{\chi_c(G)})$.
\end{proposition}

\begin{proof}
A $2$-dimensional $\eps$-orthonormal representation of a graph $G=(V,E)$ assigns a unit vector $x_v \in \R^2$ to each vertex $v \in V$, such that every pair of adjacent vertices $u$ and $v$ in $G$ satisfies $|\langle x_u,x_v \rangle | \leq \eps$. We may assume that each vector $x_v$ lies in the upper half of the unit circle, by multiplying some of the vectors by $-1$ if needed. Therefore, each vector $x_v$ can be expressed by a real number $\alpha_v \in [0,1)$, such that the angle between the vectors $(1,0)$ and $x_v$ is $\alpha_v \cdot \pi$. In this language, the condition $|\langle x_u,x_v \rangle | \leq \eps$ translates to $|\cos( \pi \cdot (\alpha_u-\alpha_v)) | \leq \eps$, or equivalently, \[\frac{\arccos(\eps)}{\pi} \leq |\alpha_u-\alpha_v| \leq 1- \frac{\arccos(\eps)}{\pi}.\]
By the definition of circular chromatic number, such a mapping $v \mapsto \alpha_v$ exists if and only if it holds that $\frac{\arccos(\eps)}{\pi} \leq \frac{1}{\chi_c(G)}$, that is, $\eps \geq \cos(\frac{\pi}{\chi_c(G)})$. The proof is complete.
\end{proof}

As a concluding remark, we raise the question of determining the $\eps$-orthogonality dimension of Kneser graphs.
The Kneser graph $K(n,k)$, defined for integers $n$ and $k$ with $n \geq 2k$, has vertices corresponding to all $k$-subsets of $[n]$ and edges between disjoint sets.
Settling a conjecture of Kneser~\cite{Kneser55}, Lov{\'{a}}sz~\cite{LovaszKneser} proved that $\chi(K(n,k))=n-2k+2$ as an application of the Borsuk--Ulam theorem from algebraic topology. This result has been strengthened in various ways over time. For example, it was shown by  Chen~\cite{Chen11} that the chromatic number of $K(n,k)$ coincides with its circular chromatic number, resolving a conjecture of Johnson, Holroyd, and Stahl~\cite{JHS97} (see also~\cite{ChangLZ13,LiuZhu15}). By Proposition~\ref{prop:circular}, this result characterizes the values of $\eps \in [0,1)$ for which $\od_\eps(K(n,k)) \leq 2$ holds. A more recent result~\cite{Haviv19topo,AlishahiM21} asserts that the chromatic number of $K(n,k)$ coincides with its standard orthogonality dimension (where $\eps =0$). It would thus be intriguing to determine the quantities $\od_\eps(K(n,k))$ for general parameter choices.

\section{Hardness Results}\label{sec:hard}

In this section, we prove our hardness results for low-rank matrix completion and for related problems.
The starting point of our hardness proofs is the gap coloring problem, defined as follows.
\begin{definition}[The $(k_1,k_2)$-Coloring Problem]
For positive integers $k_1 < k_2$, the {\em $(k_1,k_2)$-Coloring} problem asks to decide whether an input graph $G$ satisfies $\chi(G) \leq k_1$ or $\chi(G) \geq k_2$.
\end{definition}
We rely on the following hardness result, proved by Krokhin et al.~\cite{KrokhinOWZ23}.
Recall that the function $b: \N \rightarrow \N$ is defined by $b(n) = \binom{n}{\lfloor n/2 \rfloor}$.

\begin{theorem}[\cite{KrokhinOWZ23}]\label{thm:ColHard}
For every integer $k \geq 4$, the $(k,b(k))$-Coloring problem is $\NP$-hard.
\end{theorem}

\begin{remark}\label{remark:b(n)}
It is well known that asymptotically, $b(n) \sim \frac{2^n}{\sqrt{\pi n/2}}$.
Therefore, for every sufficiently large integer $n$, we have, say, $b(n) \geq 2^{n/2}$, hence $b(b(n)) \geq 2^{b(n)/2} \geq 2^{2^{n/2-1}}$.
Further, for every sufficiently large integer $n$, it holds that, say, $b(n) \geq \frac{2^{n-1}}{\sqrt{n}}$, and consequently, $b(b(n)) \geq 2^{b(n)/2} \geq 2^{2^{n-2}/\sqrt{n}}$.
\end{remark}

In what follows, we reduce the gap coloring problem to an intermediate problem, termed Graph Fitness, and establish its hardness via Theorem~\ref{thm:ColHard}.
While the definition of the Graph Fitness problem may appear somewhat artificial, its hardness enables us to derive all our hardness results in a unified framework.
A map of the reductions presented in this section is given in Figure~\ref{fig:reductions}.

\begin{figure}[ht]
    \centering
    \begin{tikzpicture}[every node/.style={align=center}, node distance=2.0cm]

        % Nodes (no rectangles, normal font size)
        \node (coloring) {Coloring};
        \node (graphfit) [right=of coloring, xshift=1.2cm] {Graph-Fitness};  % Adjusted distance for the arrow
        \node (orthodim) [above right=of graphfit, node distance=0.1cm] {Ortho-Dim};  % Shorter distance here
        \node (psd) [right=of graphfit, xshift=1.2cm] {PSD-Completion};  % Adjusted distance for the arrow
        \node (completion) [below right=of graphfit] {Completion};

        % Arrows with Lemma labels and styles
        \draw[->, thick] (coloring) -- (graphfit) node[midway, above] {Lemma~\ref{lemma:Col->GF}};
        \draw[->, thick] (coloring) -- (graphfit) node[midway, below] {Lemma~\ref{lemma:generalReduction}}; % Updated Lemma reference
        \draw[->, thick] (graphfit) -- (orthodim) node[midway, left, xshift=+0.2cm, yshift=0.3cm] {Lemma~\ref{lemma:Fit->OD}}; % Updated Lemma 2 reference
        \draw[->, thick] (graphfit) -- (psd) node[midway, above] {Lemma~\ref{lemma:Fit->Completion}}; % Updated Lemma 3 reference
        \draw[->, thick] (graphfit) -- (completion) node[midway, below, xshift=-1.0cm, yshift=0.0cm] {Lemma~\ref{lemma:Fit->Completion}}; % Updated Lemma 4 reference
    \end{tikzpicture}
    \caption{Reductions map.}
    \label{fig:reductions}
\end{figure}
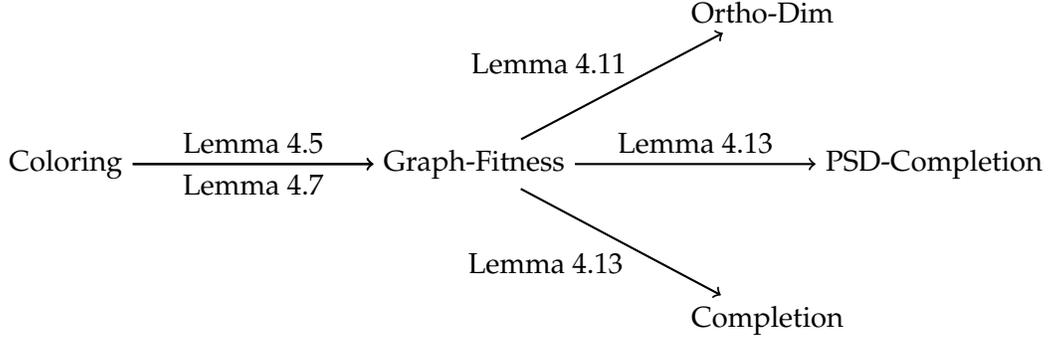

\subsection{Graph Fitness}

The Graph Fitness problem is defined as follows.

\begin{definition}[The $(d_1,d_2,\eps,\theta)$-Graph-Fitness Problem]
For positive integers $d_1 < d_2$ and real numbers $\eps \in [0,1)$ and $\theta \geq 1$, the {\em $(d_1,d_2,\eps,\theta)$-Graph-Fitness} problem asks, given a graph $G$ on $n$ vertices, to distinguish between the following cases.
\begin{itemize}
  \item $\YES$: There exists a positive semi-definite matrix $B \in \R^{n \times n}$ that fits the graph $G$, such that $\mu(B)=1$ and $\rank(B) \leq d_1$.
  \item $\NO$: For any two matrices $X,Y \in \R^{n \times d}$ whose rows have norm at most $\theta$ and for which $X \cdot Y^t$ is a symmetric matrix that $\eps$-fits the graph $G$, it holds that $d \geq d_2$.
\end{itemize}
\end{definition}
\noindent
Note that the condition on $\YES$ instances in the above definition implies the existence of a matrix $X \in \R^{n \times d_1}$ whose rows have norm $1$, such that $X \cdot X^t$ is a symmetric matrix that fits the graph $G$. Therefore, the $\YES$ and $\NO$ instances of the problem do not overlap.

We present two efficient reductions from the $(k_1,k_2)$-Coloring problem to the $(d_1,d_2,\eps,\theta)$-Graph-Fitness problem for suitable choices of parameters.
The first builds on Theorem~\ref{thm:chi(G)}.

\begin{lemma}\label{lemma:Col->GF}
Let $d_1 < d_2$ be positive integers, and let $\eps \in [0,1)$ and $\theta \geq 1$ be real numbers.
Then there exists a polynomial-time reduction from $(d_1,(\frac{4\theta^2}{1-\eps}+1)^{d_2})$-Coloring to $(d_1,d_2,\eps,\theta)$-Graph-Fitness.
\end{lemma}

\begin{proof}
Fix integers $d_1,d_2$ and real numbers $\eps,\theta$ as in the statement of the lemma.
Consider the reduction that, given an input graph $G$, produces and returns the graph $H=(V_H,E_H)$, defined as a disjoint union of $d_1$ copies of $G$.
This reduction can clearly be implemented in polynomial time (in fact, in logarithmic space).

We turn to the correctness proof of the reduction. Suppose first that $G$ is a $\YES$ instance of $(d_1,(\frac{4\theta^2}{1-\eps}+1)^{d_2})$-Coloring, that is, $\chi(G) \leq d_1$.
Since $H$ is a disjoint union of $d_1$ copies of $G$, we can apply Lemma~\ref{lemma:copies} to obtain that there exists a positive semi-definite matrix $B \in \R^{|V_H| \times |V_H|}$ that fits the graph $H$, such that $\mu(B)=1$ and $\rank(B) = d_1$. Thus, $H$ is a $\YES$ instance of $(d_1,d_2,\eps,\theta)$-Graph-Fitness.
For the converse direction, suppose that $G$ is a $\NO$ instance of $(d_1,(\frac{4\theta^2}{1-\eps}+1)^{d_2})$-Coloring, that is, $\chi(G) \geq (\frac{4\theta^2}{1-\eps}+1)^{d_2}$. For a positive integer $d$, let $X,Y \in \R^{|V_H| \times d}$ be two matrices whose rows have norm at most $\theta$, such that $X \cdot Y^t$ is a symmetric matrix that $\eps$-fits the graph $H$. By Theorem~\ref{thm:chi(G)}, we obtain that $\chi(H) \leq (\frac{4\theta^2}{1-\eps}+1)^d$. Since $\chi(G)=\chi(H)$, it follows that $d \geq d_2$, hence $H$ is a $\NO$ instance of $(d_1,d_2,\eps,\theta)$-Graph-Fitness, and we are done.
\end{proof}

\begin{remark}\label{remark:UGC_GF}
It was proved in~\cite{DinurMR06,DinurS10} that certain variants of the Unique Games Conjecture imply the hardness of the $(k_1,k_2)$-Coloring problem for all integers $k_2 > k_1 \geq 3$. By Lemma~\ref{lemma:Col->GF}, it follows that the same complexity assumptions imply the hardness of the $(d_1,d_2,\eps,\theta)$-Graph-Fitness problem for all integers $d_2>d_1 \geq 3$ and real numbers $\eps \in [0,1)$ and $\theta \geq 1$.
\end{remark}

The next reduction between the problems relies on Theorem~\ref{thm:chi(H)} and is crucial for deriving our $\NP$-hardness results from Theorem~\ref{thm:ColHard}. Note that the reduction from Lemma~\ref{lemma:Col->GF} is insufficient for this purpose. The statement involves the quantities $m(d,\eps)$ from Definition~\ref{def:m(r,eps)} and the function $b: \N \rightarrow \N$ from Theorem~\ref{thm:ColHard}.

\begin{lemma}\label{lemma:generalReduction}
Let $k_1 < k_2$ and $d_1 < d_2$ be positive integers, and let $\eps \in [0,\frac{1}{2})$, $\theta \geq 1$, and $\eta$  be real numbers, such that
\[\eta \in \bigg (0, \frac{1-2\eps}{4\theta} \bigg ],~~k_1 \leq b(d_1),~~\mbox{and}~~k_2 \geq \bigg ( \frac{2 \theta}{\eta}+1 \bigg ) ^{d_2 \cdot m(d_2,2\eta\theta+\eps)}.\]
Then there exists a polynomial-time reduction from $(k_1,k_2)$-Coloring to $(d_1,d_2,\eps,\theta)$-Graph-Fitness.
\end{lemma}

\begin{proof}
Fix integers $k_1,k_2,d_1,d_2$ and real numbers $\eps,\theta,\eta$ as in the statement of the lemma.
Consider the reduction that, given an input graph $G$, produces the underlying graph $\tilde{\delta}G$ of the digraph $\delta G$ (see Definition~\ref{def:line}), and returns the graph $H=(V_H,E_H)$, defined as a disjoint union of $d_1$ copies of $\tilde{\delta}G$.
This reduction can clearly be implemented in polynomial time (in fact, in logarithmic space).

We now prove the correctness of the reduction.
Suppose first that $G$ is a $\YES$ instance of $(k_1,k_2)$-Coloring, that is, $\chi(G) \leq k_1$.
By the assumption $k_1 \leq b(d_1)$, this implies that $\chi(G) \leq b(d_1)$, hence by Theorem~\ref{thm:chi_delta}, it follows that $\chi(\tilde{\delta}G) \leq d_1$. Since $H$ is a disjoint union of $d_1$ copies of $\tilde{\delta}G$, we can apply Lemma~\ref{lemma:copies} to obtain that there exists a positive semi-definite matrix $B \in \R^{|V_H| \times |V_H|}$ that fits the graph $H$, such that $\mu(B)=1$ and $\rank(B) = d_1$. Therefore, $H$ is a $\YES$ instance of $(d_1,d_2,\eps,\theta)$-Graph-Fitness, as desired.

For the other direction, we prove the contrapositive, namely, that if $H$ is not a $\NO$ instance of $(d_1,d_2,\eps,\theta)$-Graph-Fitness, then $G$ is not a $\NO$ instance of $(k_1,k_2)$-Coloring. To this end, suppose that for some integer $d < d_2$, there exist two matrices $X',Y' \in \R^{|V_{H}| \times d}$ whose rows have norm at most $\theta$, such that $B' = X' \cdot (Y')^t$ is a symmetric matrix that $\eps$-fits the graph $H$.
Let $X$ and $Y$ denote the sub-matrices of $X'$ and $Y'$, respectively, consisting of the rows associated with the vertices of a single copy of $\tilde{\delta}G$ in $H$.
Clearly, every row of $X$ and $Y$ has norm at most $\theta$, and $X \cdot Y^t$ is a symmetric matrix that $\eps$-fits the graph $\tilde{\delta}G$.
The assumption $\eta \in (0, \frac{1-2\eps}{4\theta}]$ allows us to apply Theorem~\ref{thm:chi(H)}, from which it follows that
\[\chi(G) \leq \bigg (\frac{2\theta}{\eta}+1 \bigg )^{d \cdot m(d,2\eta \theta+\eps)} < \bigg (\frac{2\theta}{\eta}+1 \bigg )^{d_2 \cdot m(d_2,2\eta \theta+\eps)} \leq k_2.\]
This implies that $G$ is not a $\NO$ instance of $(k_1,k_2)$-Coloring, and we are done.
\end{proof}

We next state an $\NP$-hardness result for the Graph Fitness problem.
The (somewhat tedious) proof integrates the hardness of the gap coloring problem from Theorem~\ref{thm:ColHard}, the reduction to Graph Fitness provided by Lemma~\ref{lemma:generalReduction}, and the bounds on the quantities $m(d,\eps)$ from Corollary~\ref{cor:m(d,eps)}.

\begin{theorem}\label{thm:GF_hard}
There exists an absolute constant $c>0$ for which the following holds.
Let $d$ and $g$ be positive integers with $d$ sufficiently large and $d < g$, and let $\eps$ and $\theta \geq 1$ be real numbers.
Suppose that either
\begin{enumerate}
  \item\label{itm1:Fit-Hard} $\eps \in [0, \frac{1}{3\sqrt{g}}]$ and $g \leq c \cdot \frac{2^{d/2}}{d^{1/4} \cdot  \max(\log \theta,d)^{1/2}}$, or
  \item\label{itm2:Fit-Hard} $\eps \in [ \frac{1}{3\sqrt{g}},\frac{1}{6}]$, $\theta \leq 2^{2^{c \cdot d}}$, and $g \leq c \cdot \frac{d}{\eps^2 \cdot \log (1/\eps)}$.
\end{enumerate}
Then the $(d,g,\eps,\theta)$-Graph-Fitness problem is $\NP$-hard.
\end{theorem}

\begin{proof}
Let $d \geq 4$ be an integer. For some absolute constant $c$ to be determined later, fix $g$, $\eps$, and $\theta$ as in the statement of the theorem.
We will assume, whenever needed, that $d$ (and thus $g$) is sufficiently large.
Set $k_1 = b(d)$ and $k_2 = b(b(d))$.
By Theorem~\ref{thm:ColHard}, the $(k_1,k_2)$-Coloring problem is $\NP$-hard.
It thus suffices to show that there exists a polynomial-time reduction from the $(k_1,k_2)$-Coloring problem to the $(d,g,\eps,\theta)$-Graph-Fitness problem.
We deduce the desired reduction from Lemma~\ref{lemma:generalReduction}.
To do so, we address each item of the theorem separately, as described below.

For Item~\ref{itm1:Fit-Hard}, suppose that $\eps \in [0, \frac{1}{3\sqrt{g}}]$ and $g \leq c \cdot \frac{2^{d/2}}{d^{1/4} \cdot  \max(\log \theta,d)^{1/2}}$, set $\eta = \frac{1}{6 \theta \cdot \sqrt{g}}$, and let us verify the three conditions of Lemma~\ref{lemma:generalReduction}. It is evident that $\eta \in (0, \frac{1-2\eps}{4\theta}]$ and $k_1 \leq b(d)$. It remains to verify the third condition, namely,
\begin{eqnarray}\label{eq:k_2>=}
k_2 \geq \bigg ( \frac{2 \theta}{\eta}+1 \bigg ) ^{g \cdot m(g,2\eta\theta+\eps)}.
\end{eqnarray}
Notice that $2\eta \theta +\eps = \frac{1}{3\sqrt{g}} + \eps \leq \frac{2}{3\sqrt{g}} \leq \frac{1}{\sqrt{2g}}$. This allows us to apply Item~\ref{CorAlonItm:1} of Corollary~\ref{cor:m(d,eps)}, which yields that $m(g,2\eta \theta+\eps) < 2g$. Therefore, there exists an absolute constant $c'$ such that
\[\bigg (\frac{2\theta}{\eta}+1 \bigg )^{g \cdot m(g,2 \eta \theta +\eps)} \leq (12 \theta^2 \sqrt{g}+1)^{2g^2} \leq 2^{c' \cdot (\log \theta+\log g) \cdot g^2} \leq 2^{2^{d-2}/\sqrt{d}},\]
where the last inequality follows from the assumption $g \leq c \cdot \frac{2^{d/2}}{d^{1/4} \cdot  \max(\log \theta,d)^{1/2}}$ for a sufficiently small constant $c$, which in particular gives that $\log g \leq d$.
By Remark~\ref{remark:b(n)}, this implies that for every sufficiently large $d$, the expression above does not exceed $k_2 = b(b(d))$, confirming the inequality given in~\eqref{eq:k_2>=}.

For Item~\ref{itm2:Fit-Hard}, assume that $\eps \in [ \frac{1}{3\sqrt{g}},\frac{1}{6}]$, $\theta \leq 2^{2^{c \cdot d}}$, and $g \leq c \cdot \frac{d}{\eps^2 \cdot \log (1/\eps)}$, and set $\eta = \frac{\eps}{\theta}$.
As before, it is easy to see that $\eta \in (0,\frac{1-2\eps}{4\theta}]$ and $k_1 \leq b(d)$, so it remains to verify the inequality given in~\eqref{eq:k_2>=}.
Notice that $2\eta \theta +\eps = 3\eps \in [\frac{1}{\sqrt{g}},\frac{1}{2}]$.
This allows us to apply Item~\ref{CorAlonItm:2} of Corollary~\ref{cor:m(d,eps)}, which yields that for some absolute constant $c'$, it holds that
\[ m(g,3\eps) \leq 2^{c' \cdot g \cdot \eps^2 \cdot \log (1/\eps)} \leq 2^{c' \cdot c \cdot d}.\]
This implies that for some absolute constant $c''$, it holds that
\begin{eqnarray}\label{eq:size}
\bigg ( \frac{2\theta}{\eta}+1 \bigg ) ^{g \cdot m(g,2\eta \theta+\eps)} = \bigg ( \frac{2\theta^2}{\eps}+1 \bigg ) ^{g \cdot m(g,3\eps)} \leq (6\theta^2 \sqrt{g}+1)^{g \cdot m(g,3\eps)} \leq 2^{c'' \cdot (\log \theta+\log g) \cdot g \cdot 2^{c' \cdot c \cdot d}}.
\end{eqnarray}
To bound the above expression, we observe that for some absolute constant $c'''$, it holds that
\begin{eqnarray}\label{eq:g,log(1/eps)}
g \leq 2^{c''' \cdot g \cdot \eps^2 \cdot \log(1/\eps)} \leq 2^{c''' \cdot c \cdot d}
\end{eqnarray}
for any sufficiently large $g$.
Indeed, the first inequality can be verified for a sufficiently large $c'''$ by considering the two cases $\eps \in [\frac{1}{3\sqrt{g}},\frac{1}{g^{0.4}}]$ and $\eps \in [ \frac{1}{g^{0.4}},\frac{1}{6}]$, and the second follows from the assumption $g \leq c \cdot \frac{d}{\eps^2 \cdot \log (1/\eps)}$.
By combining~\eqref{eq:g,log(1/eps)} with our assumption $\theta \leq 2^{2^{c \cdot d}}$, we obtain that the exponent of the right-hand side of~\eqref{eq:size} is bounded from above by
\[c'' \cdot (2^{c \cdot d}+c''' \cdot c \cdot d) \cdot 2^{c''' \cdot c \cdot d} \cdot 2^{c' \cdot c \cdot d} \leq 2^{d/2-1},\]
where the inequality holds for every sufficiently large $d$, by choosing $c$ as a sufficiently small constant.
This implies that the right-hand side of~\eqref{eq:size} does not exceed $2^{2^{d/2-1}}$, which, by Remark~\ref{remark:b(n)}, is bounded above by $k_2 = b(b(d))$ for any sufficiently large $d$.
This yields the inequality~\eqref{eq:k_2>=} and completes the proof.
\end{proof}

Theorem~\ref{thm:GF_hard} establishes the $\NP$-hardness of the Graph Fitness problem for general parameter settings. To illustrate its applicability, we present three implications below. The first provides an exponential gap in the rank for an exponentially small error, the second offers a polynomial gap in the rank along with a polynomial decay of the error, and the third shows a constant multiplicative gap in the rank for a constant error.

\begin{theorem}\label{thm:GF_cases}
There exists an absolute constant $c>0$ for which the following holds.
\begin{enumerate}
  \item\label{itm:GF1} There exists an absolute constant $c'>0$, such that for every sufficiently large positive integer $d$, the $(d,c \cdot \frac{2^{d/2}}{d^{3/4}},2^{-c' \cdot d},2^d)$-Graph-Fitness problem is $\NP$-hard.
  \item\label{itm:GF2} For every $\beta > 1$, there exists some $c'>0$, such that for every sufficiently large positive integer $d$, the $(d, d^\beta, c' \cdot \frac{1}{(d^{\beta-1} \cdot \log d)^{1/2}},2^{2^{c \cdot d}})$-Graph-Fitness problem is $\NP$-hard.
  \item\label{itm:GF3} For every $\alpha > 1$, there exists some $\eps \in (0,1)$, such that for every sufficiently large positive integer $d$, the $(d,\alpha \cdot d,\eps,2^{2^{c \cdot d}})$-Graph-Fitness problem is $\NP$-hard.
\end{enumerate}
\end{theorem}

\begin{proof}
Let $c$ be the constant given in Theorem~\ref{thm:GF_hard}, and let $d$ be a positive integer.
We derive the three items from Theorem~\ref{thm:GF_hard}.
For Item~\ref{itm:GF1}, apply Item~\ref{itm1:Fit-Hard} of Theorem~\ref{thm:GF_hard} with $\theta = 2^{d}$, $g = c \cdot \frac{2^{d/2}}{d^{3/4}}$, and $\eps = \frac{1}{3\sqrt{g}}$.
Next, for Item~\ref{itm:GF2}, consider some $\beta > 1$, set $\eps =c' \cdot \frac{1}{(d^{\beta-1} \cdot \log d)^{1/2}}$ and $g = c \cdot \frac{d}{\eps^2 \cdot \log (1/\eps)}$, where $c'$ is sufficiently small, ensuring that $g \geq d^\beta$.
Notice that for a sufficiently large $d$, it holds that $\eps \in [\frac{1}{3\sqrt{g}},\frac{1}{6}]$, allowing us to apply Item~\ref{itm2:Fit-Hard} of Theorem~\ref{thm:GF_hard} with $\eps$, $\theta = 2^{2^{c \cdot d}}$, and $g$, to derive the desired hardness result.
Finally, for Item~\ref{itm:GF3}, consider some $\alpha > 1$, let $\eps \in(0,\frac{1}{6}]$ be the largest real number satisfying $\alpha \leq c \cdot \frac{1}{\eps^2 \cdot \log(1/\eps)}$, and set $g = \alpha \cdot d \leq c \cdot \frac{d}{\eps^2 \cdot \log(1/\eps)}$.
Since $\eps \in [\frac{1}{3\sqrt{g}},\frac{1}{6}]$ for any sufficiently large $d$, we can apply again Item~\ref{itm2:Fit-Hard} of Theorem~\ref{thm:GF_hard} with $\eps$, $\theta = 2^{2^{c \cdot d}}$, and $g$, so the proof is complete.
\end{proof}

\subsection{Orthogonality Dimension}

We now turn to proving $\NP$-hardness results for approximating the $\eps$-orthogonality dimension of graphs, extending a result of~\cite{ChawinH23}. Consider the following gap problem (recall Definition~\ref{def:od}).

\begin{definition}[The $(d_1,d_2,\eps)$-Ortho-Dim Problem]
For positive integers $d_1 < d_2$ and a real number $\eps \in [0,1)$, the {\em $(d_1,d_2,\eps)$-Ortho-Dim} problem asks to decide whether an input graph $G$ satisfies $\od(G) \leq d_1$ or  $\od_\eps(G) \geq d_2$.
\end{definition}
\noindent
Note that the definition requires $\YES$ instances to have a bounded standard orthogonality dimension (with $\eps=0$), making our hardness results stronger.

The following lemma provides a reduction from the Graph Fitness problem to the Orthogonality Dimension problem.
\begin{lemma}\label{lemma:Fit->OD}
For all positive integers $d_1 < d_2$ and real numbers $\eps \in [0,1)$, there exists a polynomial-time reduction from $(d_1,d_2,\eps,1)$-Graph-Fitness to $(d_1,d_2,\eps)$-Ortho-Dim.
\end{lemma}

\begin{proof}
We claim that the identity reduction satisfies the statement of the lemma.
To see this, fix integers $d_1 < d_2$ and a real number $\eps \in [0,1)$, and consider an input graph $G=(V,E)$ on $n$ vertices.

If $G$ is a $\YES$ instance of $(d_1,d_2,\eps,1)$-Graph-Fitness, then there exists a positive semi-definite matrix $B \in \R^{n \times n}$ that fits the graph $G$ and satisfies $\rank(B) \leq d_1$. Since $B$ is positive semi-definite, there exists a matrix $X \in \R^{n \times d_1}$ for which $B = X \cdot X^t$. To each vertex $v \in V$, assign the $d_1$-dimensional row $x_v$ of $X$ associated with $v$. Since $B$ fits $G$, it follows that $\|x_v\|=1$ for every $v$, and that $\langle x_u,x_v \rangle = 0$ whenever $u$ and $v$ are adjacent in $G$, thus $(x_v)_{v \in V}$ forms a $d_1$-dimensional orthonormal representation of $G$. This implies that $\od(G) \leq d_1$, so $G$ is a $\YES$ instance of $(d_1,d_2,\eps)$-Ortho-Dim.

For the converse, we show the contrapositive. Suppose that $G$ is not a $\NO$ instance of $(d_1,d_2,\eps)$-Ortho-Dim, hence for some integer $d < d_2$, $G$ admits a $d$-dimensional $\eps$-orthonormal representation $(x_v)_{v \in V}$. Letting $X \in \R^{n \times d}$ denote the matrix in which the row associated with a vertex $v$ is the vector $x_v$, we obtain that every row of $X$ has norm $1$ and that $X \cdot X^t$ is a symmetric matrix that $\eps$-fits the graph $G$. By $d < d_2$, it follows that $G$ is not a $\NO$ instance of $(d_1,d_2,\eps,1)$-Graph-Fitness, and we are done.
\end{proof}

By combining Theorem~\ref{thm:GF_hard} with Lemma~\ref{lemma:Fit->OD}, we derive the following.

\begin{corollary}\label{cor:OD_hard}
There exists an absolute constant $c>0$ for which the following holds.
Let $d$ and $g$ be positive integers with $d$ sufficiently large and $d < g$, and let $\eps$ be a real number.
Suppose that either
\begin{enumerate}
  \item $\eps \in [0, \frac{1}{3\sqrt{g}}]$ and $g \leq c \cdot \frac{2^{d/2}}{d^{3/4}}$, or
  \item $\eps \in [ \frac{1}{3\sqrt{g}},\frac{1}{6}]$ and $g \leq c \cdot \frac{d}{\eps^2 \cdot \log (1/\eps)}$.
\end{enumerate}
Then the $(d,g,\eps)$-Ortho-Dim problem is $\NP$-hard.
\end{corollary}
\noindent
More concrete consequences can be obtained by combining Theorem~\ref{thm:GF_cases} with Lemma~\ref{lemma:Fit->OD}.

\subsection{Low-Rank Matrix Completion}

We finally turn to proving our hardness results for the low-rank matrix completion problems described in Definitions~\ref{def:PSD_prob} and~\ref{def:theta_prob}, thereby strengthening the $\NP$-hardness results of~\cite{HMRW14}. This will be accomplished through the reductions outlined in the following lemma.

\begin{lemma}\label{lemma:Fit->Completion}
For all positive integers $d_1 < d_2$ and real numbers $\eps \in [0,1)$, the following holds.
\begin{enumerate}
  \item\label{itm1:lemma_reduction} There exists a polynomial-time reduction from $(d_1,d_2,\eps,1)$-Graph-Fitness to \[ \Big (d_1,d_2,\frac{\eps}{1+\eps} \Big )\mbox{-PSD-Completion}.\]
  \item\label{itm2:lemma_reduction} For every real number $\theta \geq (1+\eps)^{1/2} \cdot d_2^{1/4}$, if $d_1 < \lfloor \frac{d_2}{2} \rfloor$ then there exists a polynomial-time reduction from $(d_1,d_2,\eps,\theta)$-Graph-Fitness to
  \[\bigg (d_1,\bigg \lfloor \frac{d_2}{2} \bigg \rfloor,\frac{\eps}{1+\eps}, \frac{\theta^2}{(1+\eps) \cdot d_2^{1/2}} \bigg )\mbox{-Completion}.\]
\end{enumerate}
\end{lemma}

\begin{proof}
Fix integers $d_1,d_2$ and real numbers $\eps,\theta$ as in the statement of the lemma, and set
\[\eps' = \frac{\eps}{1+\eps}~\mbox{~~~and~~~}\theta' = \frac{\theta^2}{(1+\eps) \cdot d_2^{1/2}}.\]
Note that $\eps' \in [0,1)$ and $\theta' \geq 1$.
Both parts of the lemma are established through the same reduction, which is described next.
For a given graph $G=(V,E)$ on $n$ vertices, the reduction produces and returns the partial matrix $A \in \{0,1,\perp\}^{n \times n}$, whose rows and columns are indexed by $V$, defined by
\[
A_{u,v} =
\begin{cases}
1 & \text{if } u = v, \\
0 & \text{if } \{u, v\} \in E, \\
\perp & \text{otherwise.}
\end{cases}
\]
This reduction can clearly be implemented in polynomial time (in fact, in logarithmic space).

We now prove the correctness of the reduction. We begin with the forward direction, which applies to both parts of the lemma.
Suppose that $G$ is a $\YES$ instance of $(d_1,d_2,\eps,\theta)$-Graph-Fitness for some $\theta \geq 1$.
Then, there exists a positive semi-definite matrix $B \in \R^{n \times n}$ that fits the graph $G$, such that $\mu(B)=1$ and $\rank(B) \leq d_1$.
Since $B$ fits $G$, the diagonal entries of $B$ are all ones, and the entries of $B$ associated with adjacent vertices are zeros. This implies that $A_{u,v} = B_{u,v}$ whenever $A_{u,v} \neq \perp$. It follows that $A$ is a $\YES$ instance of both $(d_1,d_2,\eps')$-PSD-Completion and $(d_1,\lfloor \frac{d_2}{2} \rfloor,\eps',\theta')$-Completion, as needed.

For the reverse direction, we prove the contrapositive.
For Item~\ref{itm1:lemma_reduction} of the lemma, we show that if $A$ is not a $\NO$ instance of $(d_1,d_2,\eps')$-PSD-Completion, then $G$ is not a $\NO$ instance of $(d_1,d_2,\eps,1)$-Graph-Fitness.
Suppose that for some integer $d < d_2$, there exists a positive semi-definite matrix $B \in \R^{n \times n}$ of rank $d$, such that $|A_{u,v} - B_{u,v}| \leq \eps'$ whenever $A_{u,v} \neq \perp$.
Since $B$ is positive semi-definite and of rank $d$, one may write $B = X \cdot X^t$ for a matrix $X \in \R^{n \times d}$. For each vertex $v \in V$, let $x_v$ denote the row of $X$ associated with $v$. Notice that for every vertex $v \in V$, it holds that $\|x_v\|^2 = \langle x_v,x_v \rangle =B_{v,v} \in [1-\eps',1+\eps']$, and that for every pair of adjacent vertices $u$ and $v$ in $G$, it holds that $\langle x_u,x_v \rangle =B_{u,v} \in [-\eps',+\eps']$. For each vertex $v \in V$, let $x'_v = \frac{x_v}{\|x_v\|}$.
It follows that for every $v \in V$, $x'_v$ is a unit vector, and that for every pair of adjacent vertices $u$ and $v$ in $G$, it holds that
\[\langle x'_u,x'_v \rangle = \frac{\langle x_u,x_v \rangle}{\|x_u\| \cdot \|x_v\|} \in \bigg [ -\frac{\eps'}{1-\eps'}, +\frac{\eps'}{1-\eps'}  \bigg ] = [-\eps,+\eps].\]
Letting $X' \in \R^{n \times d}$ denote the matrix in which the row associated with a vertex $v$ is $x'_v$, we obtain that every row of $X'$ has norm $1$ and that $X' \cdot (X')^t$ is a symmetric matrix that $\eps$-fits the graph $G$. By $d<d_2$, it follows that $G$ is not a $\NO$ instance of $(d_1,d_2,\eps,1)$-Graph-Fitness, as desired.

For Item~\ref{itm2:lemma_reduction} of the lemma, we show that if $A$ is not a $\NO$ instance of $(d_1,\lfloor \frac{d_2}{2} \rfloor,\eps',\theta')$-Completion, then $G$ is not a $\NO$ instance of $(d_1,d_2,\eps,\theta)$-Graph-Fitness.
Suppose that for some integer $d < \lfloor \frac{d_2}{2} \rfloor$, there exists a matrix $B \in [-\theta',+\theta']^{n \times n}$ of rank $d$, such that $|A_{u,v} - B_{u,v}| \leq \eps'$ whenever $A_{u,v} \neq \perp$. Put $C = \frac{1}{2} \cdot (B+B^t)$, and notice that $C$ is a symmetric matrix that lies in $[-\theta',+\theta']^{n \times n}$ and satisfies $\rank(C) \leq 2d < d_2$. By the symmetry of $A$, we observe that for all $u,v \in V$ with $A_{u,v} \neq \perp$, it holds that
\begin{eqnarray}\label{eq:AvsC}
| A_{u,v} - C_{u,v} | = \Big |\tfrac{1}{2} \cdot (A_{u,v} - B_{u,v})+\tfrac{1}{2} \cdot (A_{v,u} - B_{v,u}) \Big | \leq \tfrac{1}{2} \cdot |A_{u,v} - B_{u,v}|+\tfrac{1}{2} \cdot |A_{v,u} - B_{v,u}| \leq \eps'.
\end{eqnarray}
By Lemma~\ref{lemma:factor}, there exist two matrices $X,Y \in \R^{n \times (2d)}$ satisfying $C = X \cdot Y^t$, such that every row of $X$ and $Y$ has norm at most
\begin{eqnarray}\label{eq:normXY}
(2d)^{1/4} \cdot \theta'^{1/2} <  d_2^{1/4} \cdot \theta'^{1/2} = \frac{\theta}{(1+\eps)^{1/2}} = (1-\eps')^{1/2} \cdot \theta.
\end{eqnarray}
For each vertex $v \in V$, let $x_v$ and $y_v$ denote the rows associated with $v$ in $X$ and $Y$, respectively.
By~\eqref{eq:AvsC}, for every vertex $v \in V$, it holds that $\langle x_v,y_v \rangle = C_{v,v} \in [1-\eps',1+\eps']$, and for every pair of adjacent vertices $u$ and $v$ in $G$, it holds that $\langle x_u,y_v \rangle =C_{u,v} \in [-\eps',+\eps']$. For each vertex $v \in V$, let $x'_v = \frac{x_v}{\langle x_v,y_v \rangle^{1/2}}$ and $y'_v = \frac{y_v}{\langle x_v,y_v \rangle^{1/2}}$, and observe using~\eqref{eq:normXY} that $\|x'_v\| \leq \theta$ and $\|y'_v\| \leq \theta$. For every $v \in V$, we have $\langle x'_v,y'_v \rangle =1$, and for every pair of adjacent vertices $u$ and $v$ in $G$, it holds that
\[\langle x'_u,y'_v \rangle = \frac{\langle x_u,y_v \rangle}{\langle x_u,y_u \rangle^{1/2} \cdot \langle x_v,y_v \rangle^{1/2}} \in \bigg [ -\frac{\eps'}{1-\eps'}, +\frac{\eps'}{1-\eps'}  \bigg ] = [-\eps,+\eps].\]
Let $X',Y' \in \R^{n \times (2d)}$ denote the matrices in which the rows associated with a vertex $v$ are $x'_v$ and $y'_v$, respectively. The above discussion implies that every row of $X'$ and $Y'$ has norm at most $\theta$ and that $X' \cdot (Y')^t$ is a symmetric matrix that $\eps$-fits the graph $G$. By $2d< d_2$, it follows that $G$ is not a $\NO$ instance of $(d_1,d_2,\eps, \theta)$-Graph-Fitness, completing the proof.
\end{proof}

By combining Theorem~\ref{thm:GF_hard} with the first item of Lemma~\ref{lemma:Fit->Completion}, we obtain the following.

\begin{corollary}\label{cor:PSD_hard}
There exists an absolute constant $c>0$ for which the following holds.
Let $d$ and $g$ be positive integers with $d$ sufficiently large and $d < g$, and let $\eps$ be a real number.
Suppose that either
\begin{enumerate}
  \item $\eps \in [0, \frac{1}{3\sqrt{g}}]$ and $g \leq c \cdot \frac{2^{d/2}}{d^{3/4}}$, or
  \item $\eps \in [ \frac{1}{3\sqrt{g}},\frac{1}{6}]$ and $g \leq c \cdot \frac{d}{\eps^2 \cdot \log (1/\eps)}$.
\end{enumerate}
Then the $(d,g,\frac{\eps}{1+\eps})$-PSD-Completion problem is $\NP$-hard.
\end{corollary}

Corollary~\ref{cor:PSD_hard} confirms Theorems~\ref{thm:IntroPSD} and~\ref{thm:IntroPSDexp}.
Indeed, the first item of Corollary~\ref{cor:PSD_hard} implies that for every sufficiently large integer $d$ and any real number $\eps \in [0,2^{-\Omega(d)}]$, the $(d,2^{(1-o(1))\cdot d/2},\eps)$-PSD-Completion problem is $\NP$-hard, as needed for Theorem~\ref{thm:IntroPSDexp}. Next, for any sufficiently large integer $d$ and any real number $\eps' \in [2^{-9cd},\frac{1}{7}]$, set $\eps = \frac{\eps'}{1-\eps'}$, and observe that $\eps \in [2^{-9cd},\frac{1}{6}]$. Set $g = c \cdot \frac{d}{\eps^2 \cdot \log (1/\eps)}$, and notice that $g \cdot \eps^2 = c \cdot \frac{d}{\log (1/\eps)} \geq \frac{1}{9}$, hence $\eps \in [ \frac{1}{3\sqrt{g}},\frac{1}{6}]$. By the second item of Corollary~\ref{cor:PSD_hard}, applied with $\eps$, we obtain the $\NP$-hardness of the $(d,g,\eps')$-PSD-Completion problem, which yields Theorem~\ref{thm:IntroPSD}.

By combining Theorem~\ref{thm:GF_hard} with the second item of Lemma~\ref{lemma:Fit->Completion}, we obtain the following.

\begin{corollary}\label{cor:Comp_hard}
There exists an absolute constant $c>0$ for which the following holds.
Let $d$ and $g$ be positive integers with $d$ sufficiently large and $d < \lfloor \frac{g}{2} \rfloor$, and let $\eps$ and $\theta \geq (1+\eps)^{1/2} \cdot g^{1/4}$ be real numbers.
Suppose that either
\begin{enumerate}
  \item $\eps \in [0, \frac{1}{3\sqrt{g}}]$ and $g \leq c \cdot \frac{2^{d/2}}{d^{1/4} \cdot  \max(\log \theta,d)^{1/2}}$, or
  \item $\eps \in [ \frac{1}{3\sqrt{g}},\frac{1}{6}]$, $\theta \leq 2^{2^{c \cdot d}}$, and $g \leq c \cdot \frac{d}{\eps^2 \cdot \log (1/\eps)}$.
\end{enumerate}
Then the $(d,\lfloor \frac{g}{2}\rfloor,\frac{\eps}{1+\eps}, \frac{\theta^2}{(1+\eps) \cdot g^{1/2}})$-Completion problem is $\NP$-hard.
\end{corollary}

Corollary~\ref{cor:Comp_hard} confirms Theorems~\ref{thm:IntroFit} and~\ref{thm:IntroFitExp}. Indeed, for any sufficiently large integer $d$ and any real numbers $\eps' \in [0,2^{-\Omega(d)}]$ and $\theta' \in [1,2^{2^{o(d)}}]$, set $\eps = \frac{\eps'}{1-\eps'}$ and $\theta = (1+\eps)^{1/2} \cdot g^{1/4} \cdot \theta'^{1/2}$ where $g = 2^{(1-o(1))\cdot d/2}$. By the first item of Corollary~\ref{cor:Comp_hard}, applied with $\eps$ and $\theta$, it follows that the $(d,\lfloor \frac{g}{2}\rfloor,\eps', \theta')$-Completion problem is $\NP$-hard, as needed for Theorem~\ref{thm:IntroFitExp}. Next, for any sufficiently large integer $d$ and any real numbers $\eps' \in [2^{-9cd},\frac{1}{7}]$ and $\theta' \in [1,2^{2^{O(d)}}]$, set $\eps = \frac{\eps'}{1-\eps'}$, $\theta = (1+\eps)^{1/2} \cdot g^{1/4} \cdot \theta'^{1/2}$, and $g = c \cdot \frac{d}{\eps^2 \cdot \log (1/\eps)}$. As above, we have $g \cdot \eps^2 = c \cdot \frac{d}{\log (1/\eps)} \geq \frac{1}{9}$, hence $\eps \in [ \frac{1}{3\sqrt{g}},\frac{1}{6}]$. This allows us to apply the second item of Corollary~\ref{cor:Comp_hard} with $\eps$ and $\theta$ to obtain the $\NP$-hardness of the $(d,\lfloor \frac{g}{2}\rfloor,\eps',\theta')$-Completion problem, which in turn yields Theorem~\ref{thm:IntroFit}.

Note that one may combine Theorem~\ref{thm:GF_cases} with Lemma~\ref{lemma:Fit->Completion} to obtain hardness results for concrete settings of the parameters. This is illustrated for the PSD-Completion problem, as follows.

\begin{corollary}\label{cor:PSD_cases}
There exists an absolute constant $c>0$ for which the following holds.
\begin{enumerate}
  \item There exists an absolute constant $c'>0$, such that for every sufficiently large positive integer $d$, the $(d,c \cdot \frac{2^{d/2}}{d^{3/4}},2^{-c' \cdot d})$-PSD-Completion problem is $\NP$-hard.
  \item For every $\beta > 1$, there exists some $c'>0$, such that for every sufficiently large positive integer $d$, the $(d, d^\beta, c' \cdot \frac{1}{(d^{\beta-1} \cdot \log d)^{1/2}})$-PSD-Completion problem is $\NP$-hard.
  \item For every $\alpha > 1$, there exists some $\eps \in (0,1)$, such that for every sufficiently large positive integer $d$, the $(d,\alpha \cdot d,\eps)$-PSD-Completion problem is $\NP$-hard.
\end{enumerate}
\end{corollary}

We conclude with two final remarks.
First, certain variants of the Unique Games Conjecture (see~\cite{DinurMR06,DinurS10}) imply the hardness of the $(d_1,d_2,\eps)$-PSD-Completion and $(d_1,d_2,\eps,\theta)$-Completion problems for all integers $d_2 > d_1 \geq 3$ and real numbers $\eps \in [0,\frac{1}{2})$ and $\theta \geq 1$. This follows by combining Remark~\ref{remark:UGC_GF} with Lemma~\ref{lemma:Fit->Completion}.
Second, all our hardness results for low-rank matrix completion problems remain valid when the input partial matrix is required to have a small fraction of missing values, say, polynomially small with respect to the dimension of the matrix. As observed in~\cite{HMRW14}, this follows by padding the input partial matrix with sufficiently many zero rows and columns, which reduces the fraction of missing values and preserves the minimum possible rank of a completion.

\section*{Acknowledgments}
We thank the anonymous reviewers for their insightful comments and suggestions that improved the presentation of this paper.

\bibliographystyle{abbrv}
\bibliography{completion}

\end{document}